\documentclass{elsarticle}

\usepackage{
amsmath,
amsthm,
amscd,
amssymb,
wasysym
}
\usepackage{fancyvrb}
\usepackage{tipa}

\DefineVerbatimEnvironment{Verbatim}{Verbatim}{fontsize=\footnotesize}

\usepackage{pgfplotstable}
\usepackage{comment}
\usepackage{tikz}
\usepackage{tikz-cd}
\usetikzlibrary{positioning,arrows,calc,patterns}
\usepackage{xspace}

\usepackage{graphicx}

\usepackage{url}
\usepackage{subcaption}

\theoremstyle{plain}
\newtheorem{lemma}{Lemma}
\newtheorem{definition}{Definition}

\newtheorem{theorem}{Theorem}

\newtheorem{question}{Question}

\newtheoremstyle{derp}
{3pt}
{3pt}
{}
{}
{\upshape}
{:}
{.5em}
{}
\theoremstyle{derp}
\newtheorem{example}{Example}

\newcommand{\Z}{\mathbb{Z}}

\newcommand{\N}{\mathbb{N}}

\newcommand\xqed[1]{%
  \leavevmode\unskip\penalty9999 \hbox{}\nobreak\hfill
  \quad\hbox{#1}}
\newcommand\qee{\xqed{$\fullmoon$}}

\newcommand{\dd}[2]{\begin{matrix} #1 \\ #2 \end{matrix}}

\newcommand{\supp}{\mathrm{supp}}

\newcommand{\near}{\mathrm{N}}
\newcommand{\far}{\mathrm{F}}

\usetikzlibrary{decorations.markings}
\tikzset{
  negate/.style={
    decoration={
      markings,
      mark= at position 0.5 with {
        \node[transform shape, yscale=0.5, xscale=2] (tempnode) {$/$};
      },
    },
    postaction={decorate},
  },
}

\newcommand{\gateE}{
\begin{tikzpicture}[scale=0.25]
\draw[black!40!white] (0,0) rectangle (1,1);
\end{tikzpicture}}
\newcommand{\gateH}{
\begin{tikzpicture}[scale=0.25]
\draw[black!40!white] (0,0) rectangle (1,1);
\draw (0,0.5) -- (1,0.5);
\end{tikzpicture}
}
\newcommand{\gateV}{
\begin{tikzpicture}[scale=0.25]
\draw[black!40!white] (0,0) rectangle (1,1);
\draw (0.5,0) -- (0.5,1);
\end{tikzpicture}}
\newcommand{\gateEN}{
\begin{tikzpicture}[scale=0.25]
\draw[black!40!white] (0,0) rectangle (1,1);
\draw (1,0.5) -- (0.5,0.5) -- (0.5,1);
\end{tikzpicture}}
\newcommand{\gateNW}{
\begin{tikzpicture}[scale=0.25]
\draw[black!40!white] (0,0) rectangle (1,1);
\draw (0,0.5) -- (0.5,0.5) -- (0.5,1);
\end{tikzpicture}}
\newcommand{\gateWS}{
\begin{tikzpicture}[scale=0.25]
\draw[black!40!white] (0,0) rectangle (1,1);
\draw (0,0.5) -- (0.5,0.5) -- (0.5,0);
\end{tikzpicture}}
\newcommand{\gateES}{
\begin{tikzpicture}[scale=0.25]
\draw[black!40!white] (0,0) rectangle (1,1);
\draw (1,0.5) -- (0.5,0.5) -- (0.5,0);
\end{tikzpicture}}
\newcommand{\gateT}{
\begin{tikzpicture}[scale=0.25]
\draw[black!40!white] (0,0) rectangle (1,1);
\draw (0,0.5) -- (0.5,0.5);
\draw (0.5,0.5) -- (0.7,0.7);
\draw (0.5,0.5) -- (0.7,0.3);
\end{tikzpicture}}
\newcommand{\gateN}{
\begin{tikzpicture}[scale=0.25]
\draw[black!40!white] (0,0) rectangle (1,1);
\draw (0,0.5) -- (1,0.5);
\filldraw[black] (0.5,0.5) circle (0.15);
\end{tikzpicture}}
\newcommand{\gateS}{
\begin{tikzpicture}[scale=0.25]
\draw[black!40!white] (0,0) rectangle (1,1);
\draw (0.5,0.5) -- (1,0.5);
\draw (0.5,0) -- (0.5,1);
\end{tikzpicture}}
\newcommand{\gateX}{
\begin{tikzpicture}[scale=0.25]
\draw[black!40!white] (0,0) rectangle (1,1);
\draw (0,0.5) -- (1,0.5) (0.5,0) -- (0.5,1);
\end{tikzpicture}}
\newcommand{\gateO}{
\begin{tikzpicture}[scale=0.25]
\draw[black!40!white] (0,0) rectangle (1,1);
\draw (0.5,1) -- (0.5,0.75);
\draw (0.5,0) -- (0.5,0.25);
\draw (0.8,0.5) -- (1,0.5);
\draw (0.25,0.2) edge[bend right=80,in=-90, out=-90, looseness = 1] (0.25,0.8);
\draw (0.25,0.2) edge[bend right=80,in=-90, out=-90, looseness = 3] (0.25,0.8);
\end{tikzpicture}}

\makeatletter
\tikzset{
    zero color/.initial=white,
    zero color/.get=\zerocol,
    zero color/.store in=\zerocol,
    one color/.initial=red,
    one color/.get=\onecol,
    one color/.store in=\onecol,
    two color/.initial=blue,
    two color/.get=\twocol,
    two color/.store in=\twocol,
    cell wd/.initial=1,
    cell wd/.get=\cellwd,
    cell wd/.store in=\cellwd,
    cell ht/.initial=1,
    cell ht/.get=\cellht,
    cell ht/.store in=\cellht,
}
\makeatother

\newcommand{\drawgrid}[2][]{
\medskip
\begin{tikzpicture}[#1]
  \pgfplotstablegetrowsof{#2} 
  \pgfmathtruncatemacro{\totrow}{\pgfplotsretval}
  \pgfplotstablegetcolsof{#2} 
  \pgfmathtruncatemacro{\totcol}{\pgfplotsretval}
  
  \pgfplotstableforeachcolumn#2\as\col{
    \pgfplotstableforeachcolumnelement{\col}\of#2\as\colcnt{%
      \ifnum\colcnt=0
        \fill[\zerocol]($ -\pgfplotstablerow*(0,\cellht) + \col*(\cellwd,0) $) rectangle+(\cellwd,\cellht);
      \fi
      \ifnum\colcnt=1
        \fill[\onecol]($ -\pgfplotstablerow*(0,\cellht) + \col*(\cellwd,0) $) rectangle+(\cellwd,\cellht);
      \fi
      \ifnum\colcnt=2
        \fill[\twocol]($ -\pgfplotstablerow*(0,\cellht) + \col*(\cellwd,0) $) rectangle+(\cellwd,\cellht);
      \fi
    }
  }
\end{tikzpicture}
\medskip
}

\newcommand\kek{
      \pgfplotstablegetrowsof{\matrixfile} 
      \pgfmathtruncatemacro{\totrow}{\pgfplotsretval}
      \pgfplotstablegetcolsof{\matrixfile} 
      \pgfmathtruncatemacro{\totcol}{\pgfplotsretval}
      
      \pgfplotstableforeachcolumn\matrixfile\as\col{
        \pgfplotstableforeachcolumnelement{\col}\of\matrixfile\as\colcnt{%
          \ifnum\colcnt=0
          \fill[white]($ -\pgfplotstablerow*(0,\cellht) + \col*(\cellwd,0) $) rectangle +(\cellwd,\cellht);
          \fi
          \ifnum\colcnt=1
          \fill[black]($ -\pgfplotstablerow*(0,\cellht) + \col*(\cellwd,0) $) rectangle+(\cellwd,\cellht);
          \fi
          \ifnum\colcnt=2
          \fill[black]($ -\pgfplotstablerow*(0,\cellht) + \col*(\cellwd,0) $) rectangle+(\cellwd,\cellht);
          \fill[white]($ -\pgfplotstablerow*(0,\cellht) + \col*(\cellwd,0) + (\cellwd/2,\cellht/2) $) circle (\cellwd/4);
          \fi
          \ifnum\colcnt=3
          \fill[black]($ -\pgfplotstablerow*(0,\cellht) + \col*(\cellwd,0) $) rectangle+(\cellwd,\cellht);
          \draw[white,thick]($ -\pgfplotstablerow*(0,\cellht) + \col*(\cellwd,0) + (\cellwd/4,\cellht/4) $) -- +(\cellwd/2,\cellht/2);
          \draw[white,thick]($ -\pgfplotstablerow*(0,\cellht) + \col*(\cellwd,0) + (\cellwd/4,3*\cellht/4) $) -- +(\cellwd/2,-\cellht/2);
          \fi
        }
      }}

\newcommand\keke{
      \pgfplotstablegetrowsof{\matrixfile} 
      \pgfmathtruncatemacro{\totrow}{\pgfplotsretval}
      \pgfplotstablegetcolsof{\matrixfile} 
      \pgfmathtruncatemacro{\totcol}{\pgfplotsretval}
      
      \pgfplotstableforeachcolumn\matrixfile\as\col{
        \pgfplotstableforeachcolumnelement{\col}\of\matrixfile\as\colcnt{%
          \ifnum\colcnt=0
          \fill[white]($ -\pgfplotstablerow*(0,\cellht) + \col*(\cellwd,0) $) rectangle +(\cellwd,\cellht);
          \fi
          \ifnum\colcnt=1
          \fill[black]($ -\pgfplotstablerow*(0,\cellht) + \col*(\cellwd,0) $) rectangle+(\cellwd,\cellht);
          \fi
          \ifnum\colcnt=2
          \fill[black]($ -\pgfplotstablerow*(0,\cellht) + \col*(\cellwd,0) $) rectangle+(\cellwd,\cellht);
          \fi
          \ifnum\colcnt=3
          \fill[black]($ -\pgfplotstablerow*(0,\cellht) + \col*(\cellwd,0) $) rectangle+(\cellwd,\cellht);
          \fi
        }
      }}

\journal{Theoretical Computer Science}

\begin{document}
\begin{frontmatter}

\title{Structure and computability of preimages in the Game of Life}

\author[1]{Ville Salo}
\ead{vosalo@utu.fi}
\author[1]{Ilkka Törmä\corref{cor}\fnref{fn}}
\ead{iatorm@utu.fi}

\cortext[cor]{Corresponding author}
\fntext[fn]{Ilkka Törmä was supported by the Academy of Finland under grant 346566.}

  
\address[1]{Department of Mathematics and Statistics, University of Turku, 20014 Turku, Finland}

\begin{abstract}
  Conway's Game of Life is a two-dimensional cellular automaton. As a dynamical system, it is well-known to be computationally universal, i.e.\ capable of simulating an arbitrary Turing machine. We show that in a sense taking a single backwards step of the Game of Life is a computationally universal process, by constructing patterns whose preimage computation encodes an arbitrary circuit-satisfaction problem, or, equivalently, any tiling problem. As a corollary, we obtain for example that the set of orphans is coNP-complete, exhibit a $6210 \times 37800$-periodic configuration whose preimage is nonempty but contains no periodic configurations, and prove that the existence of a preimage for a periodic point is undecidable. Our constructions were obtained by a combination of computer searches and manual design.
\end{abstract}


\end{frontmatter}

\section{Introduction}
\label{sec:Intro}

Conway's Game of Life, designed by John Conway and popularized by Martin Gardner \cite{Ga70}, is a two-dimensional cellular automaton. Specifically, it is the function $g : \{0,1\}^{\Z^2} \to \{0,1\}^{\Z^2}$ defined by
\begin{equation}
\label{eq:GoLdef}
g(x)_{\vec v} = 
\begin{cases}
1, & \text{if~} x_{\vec v} = 0 \text{~and~} \sum_{\vec n \in [-1,1]^2} x_{\vec v + \vec n} = 3, \\
1, & \text{if~} x_{\vec v} = 1 \text{~and~} \sum_{\vec n \in [-1,1]^2} x_{\vec v + \vec n} \in \{3,4\}, \\
0, & \text{otherwise}
\end{cases}
\end{equation}
where $x \in \{0,1\}^{\Z^2}$ and $\vec v \in \Z^2$.
If configurations $x \in \{0,1\}^{\Z^2}$ are interpreted as infinite grids of live/black/occupied cells (denoted by $1$) and dead/white/empty cells (denoted by $0$), then the rule can be interpreted as saying that a new live cell is born at a dead cell with exactly three live neighbors, and a live cell will stay alive if and only if it has two or three live neighbors.

The Game of Life, being a function from a topological (specifically, Cantor) space to itself, can be interpreted as a dynamical system. As the space is very combinatorial, it can also naturally be interpreted as a computational device, and these two points of view are strongly intertwined. Starting from a finite-support configuration (one with finitely many live cells) it can simulate the behavior of a Turing machine on a finite but unbounded configuration \cite{Re02}.\footnote{Adam Goucher obtains a simultaneous proof of both claims by proving intrinsic universality with a quiescent state simulated by $0$-blocks \cite[Chapter 12]{JoGr22}.}
The Game of Life is also intrinsically universal, i.e.\ its subsystems can simulate any cellular automaton \cite{DuRo99}.
Thus, it is natural to say that the dynamics of the Game of Life is computationally universal. See \cite{JoGr22,lifewiki} for more information on building patterns in the Game of Life with interesting dynamics, which is a rather vast field of science on its own.

In addition to its dynamics, the Game of Life can also thought of as a \emph{block map}, i.e.\ a continuous shift-commuting map $g : X \to Y$ between two subshifts $X = \{0,1\}^{\Z^2}$ and $Y = \{0,1\}^{\Z^2}$.\footnote{In symbolic dynamics, a \emph{subshift} is a certain kind of set of infinite configurations.}
Since $g$ is a cellular automaton, these subshifts are of course identical, but we can study them as separate objects with different roles.
In this point of view, our attention moves from iteration to ``one-step'' problems. For example, we may ask about the image of the Game of Life $g(\{0,1\}^{\Z^2})$ (in symbolic dynamics jargon, this is a particular \emph{sofic shift}), about its fixed points (in symbolic dynamics jargon, this is a particular \emph{subshift of finite type}, or \emph{SFT}), or about the preimage $g^{-1}(0^{\Z^2})$ of the completely empty configuration (this is another SFT).

Previous results in this direction are that the maximal density of $1$s in the subshift of fixed points is exactly $1/2$ \cite{El99}, that $g^{-1}(0^{\Z^2})$ has dense semilinear points but does not have dense periodic points \cite{SaTo22}, and that there exist finite patterns $P$ such that each pattern in $g^{-1}(P)$ also contains an occurrence of $P$ (which has various implications for the long-term dynamics of $g$) \cite{SaTo22b}. There are also many open problems about these subshifts, for example, the Still Life Finitization Problem asks whether the subshift of fixed points has dense finite-support configurations.

One may also consider the study of temporally periodic points as being about one-step dynamics (seeing higher powers of the Game of Life as block maps). An interesting problem is \emph{omniperiodicity}, or whether the Game of Life admits finite-support configurations of all periods. It was recently solved in the positive \cite{omniper}, ending a decades-long collaborative project of the Game of Life research community. It remains open what other properties the periodic-point subshifts have; for example, like with the fixed-point subshift, it is natural to ask whether finite-support points are dense.

In this paper, we prove that the Game of Life is ``universal as a block map'' (in a rather technical sense discussed in Section~\ref{sec:Universality}). In particular, we can argue that producing a preimage of a configuration is a computationally universal process, roughly equivalent to finding a satisfying assignment to a Boolean circuit. We construct ``gadgets'', square patterns of $0$s and $1$s that represent logical gates, such that the preimages of circuits built from these gadgets are in correspondence with the satisfying assignments of the corresponding abstract circuit. For example, we show below a ``wire-crossing'' gadget, which allows two signals to pass through it without interacting, and a corresponding preimage where signals cross:


  \begin{center}
    \pgfplotstableread{crossing2preimage.cvs}{\matrixfile}
      \begin{tikzpicture}[scale = 0.1]
\keke
      \draw[black!30!white] (0,-37) grid (51,1);
    \end{tikzpicture}
    \raisebox{1.8cm}{\hspace{0.1cm}$\overset{g}{\mapsto}$\hspace{-0.05cm}}
    \pgfplotstableread{crossing2.cvs}{\matrixfile}
    \begin{tikzpicture}[scale = 0.1]
\keke
      \draw[black!30!white] (0,-35) grid (49,1);
      \draw[white] (0,-36) grid (0,-35.5); 
    \end{tikzpicture}
  \end{center}

  On the right hand side is the gadget.
  The black lines of width 2 are called \emph{wires}.
  On the left hand side is a preimage of the gadget under $g$.
  The repeated length-4 segments form one particular preimage of such a wire.
  This repeating pattern has the useful property that it propagates through a wire: if we have a preimage of a wire, one segment of which contains this pattern, then the entire preimage must consist of repetitions of the pattern with period 3.
  Note that the word ``propagate'' does not imply a temporal relation, but a constraint between different parts of a single preimage.

The 3-periodic preimage of a wire can be in one of three phases, and we choose two of them -- called its \emph{operating phases} -- to represent a value of $0$ or $1$ carried by the wire.
In the above gadget, the signals of the north and west wires are correlated on the preimage, as are those of the east and west wires, as long as all four are in an operating phase.
We present gadgets for tasks like turning wires, inverting the signal they carry, and logical operations, as well as more technical gadgets that introduce the signal into wires and change the phase of the signal traveling on them.

Our general construction allows us to turn the problem of satisfying a circuit into a preimage computation problem.
It follows that given a finite-support configuration, it is NP-complete whether this configuration admits a preimage under the Game of Life.
The problem is already known to be in NP \cite{SaTo22}, and we prove NP-hardness in Theorem~\ref{thm:FinitePreimageNPC}.
Also, given a finite pattern, it is NP-complete whether it appears in the image subshift of the Game of life, equivalently whether it is an orphan is co-NP-complete (Theorem~\ref{thm:Orphans}).
We note that for a general fixed cellular automaton, the former problem may be undecidable (by an easy reduction from the halting problem: the finite-support configuration encodes an input to a fixed universal Turing machine, and the preimage must contain an infinite computation, which the automaton erases in a single step), while the latter is always in NP.

We can also simulate tilings of the two-dimensional plane by drawing circuits that check the color constraints of a set of Wang tiles. An immediate, but perhaps striking, corollary is that there exists a periodic configuration that has a preimage, but has no periodic preimage (Theorem~\ref{thm:AperiodicOnly}). Using a small aperiodic tile set, such as the 11-tile Jeandel-Rao set \cite{JeRa21}, we can make the periods small enough ($6210 \times 37800$) to fit comfortably in computer memory. In the code repository \cite{pat-repo}, the reader can find an explicit presentation compatible with the Golly simulator \cite{golly}; screenshots are shown in Figure~\ref{fig:JeandelRao}.

\begin{figure}[htp]
\begin{center}
\fbox{\includegraphics[origin=c,scale=0.5,trim={0.23cm 0.23cm 0.23cm 0.23cm}]{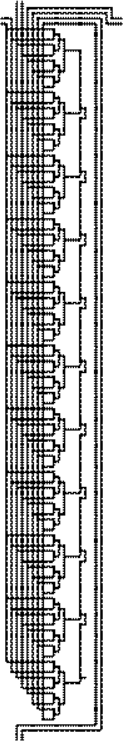}}
\hspace{4mm}
\fbox{\includegraphics[origin=c,scale=0.5,trim={0.23cm 0.23cm 0.23cm 0.23cm}]{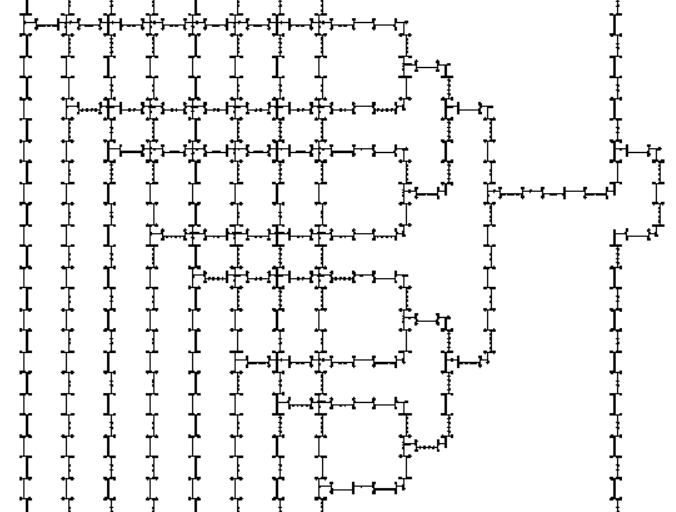}}
\end{center}
\caption{On the left, the fundamental domain (a $6210 \times 37800$ rectangle) of a periodic configuration whose preimages factor onto scaled tilings of the Jeandel-Rao tile set (thus are aperiodic). On the right, a subpattern roughly corresponding to a single tile type.}
\label{fig:JeandelRao}
\end{figure}

More strongly, we prove that there exists a periodic configuration that has a preimage, but has no recursive (= computable) preimage (Theorem~\ref{thm:ArecursiveOnly}).
We can also guarantee that every preimage of this periodic configuration has $\Omega(n)$ Kolmogorov complexity in every $n \times n$-pattern (Theorem~\ref{thm:Kolmogorov}).
On the decidability side, we show that given a periodic configuration, it is undecidable whether it admits a preimage (Theorem~\ref{thm:PeriodicPreimageUndecidable}), and given a periodic configuration that has a preimage, it is undecidable whether it has a periodic preimage (Theorem~\ref{thm:PreimagePeriodicUndecidable}).
We emphasize that no real work is needed to establish these corollaries -- we simply transport well-known results from tiling theory into the Game of Life using our circuit technology.

Most single-step properties of two-dimensional cellular automata, such as injectivity or reversibility, are undecidable, as are most dynamical properties of CA in any number of dimensions \cite{Ka18,Su18}.
By contrast, for one-dimensional CA single-step problems are much easier than for higher-dimensional CA, as they can be modeled by finite state machines \cite{Su15}.
In particular, the orphans of a one-dimensional CA form a regular language that can be explicitly computed from the local rule \cite{Wo84}.
Even the dynamical properties of one- and higher-dimensional CA tend to be qualitatively different, even though both are undecidable \cite{SaTh11}.

\medskip

The technical framework where we prove most of our results is adapted from the idea of intrinsic universality in cellular automata theory \cite{iu,DeMaOlTh11,DeMaOlTh11b}. The basic idea is that a universal block map $\phi$ can simulate any other block map $\psi$, when symbols of (the codomain of) $\psi$ are replaced by suitable rectangular patterns (in the codomain) of $\phi$. We in fact introduce three notions of universality for block maps, namely weak, semiweak and strong universality. We develop a basic theory of these notions, and show that a certain block map taking satisfied logical circuits to their underlying circuits (all encoded as subshifts of finite type) is strongly universal.

We prove that the Game of Life semiweakly simulates the circuit system, and conclude that it is semiweakly universal (Theorem~\ref{thm:GoLSemiweaklyUniversal}).
This suffices to prove all the results above.
Indeed, all but Theorem~\ref{thm:PreimagePeriodicUndecidable} follow from weak universality.

\medskip

All of our gadgets were found using SAT-solvers to find preimages (and prove the lack of particular types of preimages) for patterns, and then looking for suitable patterns using various types of search methods, in some cases accompanied by a bit of manual experimentation.
Our paper is by no means the first time ``NP-completeness gadgets'' are found by computer \cite{ErRu13,RuHo10}.\footnote{See references in \cite{NPcomp} for more examples.}

The results are conditional to SAT-solvers' claims about the gadgets being correct. We have not been able to verify any but the most obvious claims with straightforward search algorithms. Information on the implementation and testing can be found in Section~\ref{sec:Implementation}.




\section{Definitions}

Let $A$ be a finite set, called a \emph{state set} or \emph{alphabet}.
A two-dimensional \emph{pattern} over $A$ consists of a set $D \subseteq \Z^2$ and a mapping $P : D \to A$.
We usually refer to $P$ as the pattern and $D = D(P)$ as the \emph{domain} of $P$.
The symbol at position $\vec v \in D$ in $P$ is denoted by $P_{\vec v}$.
A \emph{finite} pattern is a pattern with a finite domain, and a \emph{configuration} is a pattern whose domain is $\Z^2$.
Configurations are denoted by small letters $x, y, z, \ldots$, and the set $A^{\Z^2}$ of all configurations is the two-dimensional \emph{full shift} over $A$.
If $A = \{0,1\}$ is the binary alphabet, the \emph{support} of a configuration $x$ is the set $\supp(x) = \{ \vec v \in \Z^2 \mid x_{\vec v} = 1 \}$. We also use the obvious analog for finite patterns.
The \emph{shift maps} $\sigma^{\vec v} : A^{\Z^2} \to A^{\Z^2}$ for $\vec v \in \Z^2$ are defined by $\sigma^{\vec v}(x)_{\vec w} = x_{\vec w + \vec v}$.
They can also be applied to patterns: $\sigma^{\vec v}(P)$ is the pattern $Q \in A^{D(P) - \vec v}$ with $Q_{\vec w} = P_{\vec w + \vec v}$ for all $\vec w \in D(P) - \vec{v}$.
A pattern $P$ \emph{occurs} in a configuration $x$, denoted by $P \sqsubset x$, if there exists $\vec v \in \Z^2$ with $\sigma^{\vec v}(x)|_{D(P)} = P$.

Every set $F$ of finite patterns over $A$ defines a \emph{subshift} $X \subset A^{\Z^2}$ as the set of configurations where no element of $F$ occurs.
Symbolically,
\[
  X = \{ x \in A^{\Z^2} \mid \forall P \in F : P \not\sqsubset x \}.
\]
If $F$ is finite, then $X$ is a \emph{shift of finite type} (SFT).

One way to define two-dimensional SFTs is by \emph{Wang tile sets}.
A \emph{Wang tile} over a finite set $C$ of \emph{edge colors} is a 4-tuple $t = (t_E, t_N, t_W, t_S) \in C^4$.
A set $T \subseteq C^4$ of Wang tiles defines an SFT $X_T \subseteq T^{\Z^2}$ using the forbidden patterns
\[
  t^1 t^2 \quad \text{and} \quad \begin{array}{c} t^3 \\ t^4 \end{array}
\]
for all $t^1, t^2, t^3, t^4 \in T$ with $(t^1)_E \neq (t^2)_W$ and $(t^3)_S \neq (t^4)_N$, respectively.
A Wang tile can be visualized as a unit square whose east, north, west and south edges are labeled with the four colors in this order.
A configuration of Wang tiles is in $X_T$ if the edge colors of all adjacent tiles match.

A \emph{cellular automaton} is a function $f : A^{\Z^2} \to A^{\Z^2}$ that admits a finite \emph{neighborhood} $N \subset \Z^2$ and a \emph{local rule} $F : A^N \to A$ such that $f(x)_{\vec v} = F(\sigma^{\vec v}(x)|_N)$ holds for all $x \in A^{\Z^2}$ and $\vec v \in \Z^2$. 
Once the neighborhood $N$ and local rule $F$ of a cellular automaton $f$ have been fixed, we define the preimage of a pattern $P$ over $A$ as
\[
  f^{-1}(P) = \{ Q \in A^{D(P) + N} \mid Q \sqsubset X, \forall \vec v \in D(P) : F(\sigma^{\vec v}(Q)|_N) = P_{\vec v} \}.
\]
Note that the common domain of the patterns in $f^{-1}(P)$ is usually larger than that of $P$.
We sometimes refer to individual patterns in $f^{-1}(P)$ as preimages of $P$.

The Game of Life $g : \{0,1\}^{\Z^2} \to \{0,1\}^{\Z^2}$ is the cellular automaton defined in Equation~\eqref{eq:GoLdef}, and it admits the neighborhood $N = [-1,1]^2$.
We will not use the letter $g$ for any other purpose.

Let $A$ and $B$ be alphabets and $R = [0, w-1] \times [0, h-1]$ a rectangle.
A \emph{substitution} of shape $R$ is a mapping $\tau : A \to B^R$.
It can be extended to a mapping $\tau : A^{\Z^2} \to B^{\Z^2}$ by $\tau(x)_{\vec v} = \tau(\sigma^{\vec w}(x)_{\vec 0})_{\vec u}$, where $\vec u \in R$ and $\vec w \in w \Z \times h \Z$ are the unique vectors that satisfy $\vec u + \vec w = \vec v$.
Intuitively, we apply the ``local'' map $\tau$ to every symbol of $x$, obtain a configuration of $R$-shaped patterns over $B$, and then stitch them together to obtain a configuration over $B$.

For our computational complexity results we need to fix encodings of various objects into binary words.
Each alphabet $A$ is assumed to come with a binary encoding $b : A \to \{0,1\}^+$ that is unambiguous and injective as a morphism, in the sense that the mapping from words $a_1 \cdots a_n \in A^*$ of arbitrarily length to concatenations $b(a_1) \cdots b(a_n)$ is injective.
For rectangles of the form $R = [-m, m] \times [-n, n]$, we encode patterns $P \in A^R$ as $0^{m} 1 0^{n} 1 u$, where $u$ is the concatenation of $b(P_{\vec v})$ for $\vec v \in R$ in lexicographical order.
A finite-support configuration $x \in \{0,1\}^{\Z^2}$ is encoded identically to the pattern $x|_R$, where $R \subset \Z^2$ is the smallest rectangle of the above form containing $\supp(x)$.
For a fixed universal Turing machine $U$, the \emph{Kolmogorov complexity} of a word $w \in \{0,1\}^*$ is the length of the shortest word $v \in \{0,1\}^*$ such that $U$ computes $w$ when given $v$ as input.
The Kolmogorov complexity of a rectangular pattern $P$ is defined as the Kolmogorov complexity of the concatenation of $b(P_{\vec v})$ for $\vec v \in D(P)$ in lexicographical order.
Note that the shape and position of the pattern are not encoded.

The following definitions are needed only for Section~\ref{sec:Universality} and the proofs of Section~\ref{sec:Corollaries}, which can be skipped on first reading.
Let $X \subseteq A^{\Z^2}$ and $Y \subseteq B^{\Z^2}$ be subshifts.
A \emph{block map} is a function $f : X \to Y$ that admits a finite neighborhood $N \subset \Z^2$ and a local rule $F : A^N \to B$ such that $f(x)_{\vec v} = F(\sigma^{\vec v}(x)|_N)$ holds for all $x \in X$ and $\vec v \in \Z^2$.
The image $f(X)$ of a block map is always a subshift.
Note that a cellular automaton is just a block map from a full shift to itself.
A \emph{section} of a block map $f : X \to Y$ is a block map $h : Y \to X$ such that $f(h(y)) = y$ for all $y \in Y$.
We say a block map \emph{splits} or \emph{is split} if it admits a section.
If a block map $f : X \to Y$ is bijective, it is called a \emph{(topological) conjugacy}, and then $X$ and $Y$ are \emph{conjugate} subshifts.
The inverse function of a conjugacy is always a block map.


\section{Wires}

In our constructions we use wires, which are a type of \emph{self-propagating pattern}, namely a pattern $P$ associated to a vector $\vec p \neq \vec 0$ with the property that whenever the $g$-preimage $x$ of a suitable configuration $y$ contains an occurrence of $P$ at some coordinate $\vec v$, then it contains another occurrence at $\vec v + \vec p$. In our case, $\vec p \in \{(0, \pm 3), (\pm 3,0)\}$, and $y$ is suitable if it resembles the all-$1$ configuration near the two occurrences of $P$. If $y$ is suitable near $\vec v + n \vec p$ for $n = 2, 3, \ldots$, then the pattern must occur in $x$ at these coordinates as well.
Again we note that the notion of propagation is not temporal, since we are only dealing with the set of preimages of a single configuration or pattern, but it is directional in the sense that an occurrence of a pattern at some position in a preimage implies another occurrence at a certain nearby position.



A \emph{vertical wire} is a pattern of shape $2 \times n$ consisting of $1$-symbols, and a \emph{horizontal wire} is a pattern of shape $n \times 2$ consisting of $1$-symbols. Information is transmitted along a wire in the form of a $3$-periodic pattern.
The \emph{vertical wire signals} are the three $4 \times 2$ patterns
\begin{align*}
  W_0 = {} & \dd{1&1&1&1}{0&0&0&0} & W_1 = {} & \dd{0&0&0&0}{1&1&1&1} & W_2 = {} & \dd{0&0&0&0}{0&0&0&0}
\end{align*}
Rotating them by 90 degrees counterclockwise results in the \emph{horizontal wire signals}.

Consider a configuration $x \in \{0, 1\}^{\Z^2}$ and $y = g(x)$.
If $x|_{[0,3] \times [0,1]} = W_i$ and $y|_{[1,2]^2} = \begin{smallmatrix} 1&1 \\ 1&1 \end{smallmatrix}$, then a short case analysis reveals that $x|_{[0,3] \times [1,2]} = W_{i+1}$, where the index is taken modulo 3.
For example, in the case $i = 0$ both $x_{(1,1)}$ and $x_{(2,1)}$ must have three $1$-symbols in their neighborhood, which implies $x|_{[0,3] \times \{2\}} = 1111$.
(One may also check that already a single $1111$-row in the preimage will force this periodic structure on the preimage.)
Thus each $W_i$ is self-propagating along a vertical wire in both directions $\vec p = (0, \pm 3)$, and the analogous claim holds for the horizontal signals and wires.
Note that this property holds no matter what the $1$s are surrounded with. Typically we will surround the wires with dead cells, and use a line of live cells of thickness two as the wire.

We say that a preimage of a wire is \emph{charged} if it contains a wire signal.
In Figure~\ref{fig:wire} we show a horizontal wire and its three charged preimages, plus one uncharged preimage. We only show the part of the charged preimages that is forced by the presence of a horizontal wire signal in any part of the wire.
Uncharged preimages of wires will not play a role in the remainder of this article, apart from being something we actively prevent from existing.

\begin{figure}[htp]
  \begin{center}
  \pgfplotstableread{wire.cvs}{\matrixfile}
    \begin{tikzpicture}[scale = 0.25]
\kek
      \draw[black!30!white] (0,-1) grid (12,1);
    \end{tikzpicture}
    
    \vspace{5mm}
  \pgfplotstableread{wirepre1.cvs}{\matrixfile}
    \begin{tikzpicture}[scale = 0.25]
\kek
      \draw[black!30!white] (0,-3) grid (12,1);
    \end{tikzpicture}
    \hspace{2mm}
  \pgfplotstableread{wirepre2.cvs}{\matrixfile}
    \begin{tikzpicture}[scale = 0.25]
\kek
      \draw[black!30!white] (0,-3) grid (12,1);
    \end{tikzpicture}
        \hspace{2mm}
      \pgfplotstableread{wirepre3.cvs}{\matrixfile}
    \begin{tikzpicture}[scale = 0.25]
\kek
      \draw[black!30!white] (0,-3) grid (12,1);
    \end{tikzpicture}
    
    \vspace{5mm}
      \pgfplotstableread{wirepre4.cvs}{\matrixfile}
    \begin{tikzpicture}[scale = 0.25]
\kek
      \draw[black!30!white] (0,-3) grid (14,1);
    \end{tikzpicture}
  \end{center}
  \caption{A horizontal wire and its three charged preimages, plus one uncharged preimage.}
  \label{fig:wire}
\end{figure}

We use the phase of the signal to transmit information within the preimages of a fixed image $y$.
Since each wire has three charged preimages, we can choose two of them -- called its \emph{operating phases} -- to represent the values $0$ and $1$, and leave the third one unused.
We use white dots in figures to represent the unused phase, marking those cells whose preimage should be $0$ regardless of the bit carried by the wire.
In Figure~\ref{fig:wire}, the unused phase, according to how we have marked the wire, is the rightmost one.

The choice of the unused phase depends on the wire, as different gadgets require different conventions.
As for the bits, in a vertical wire we will always interpret a charged wire as containing the logical value $0$ (resp.\ $1$) if, out of the two possible adjacent positions where the $1$s could occur, they occur on the north (resp.\ south) side, i.e.\ in the preimage of the unmarked cells we have a copy of $W_0$ (resp.\ $W_1$). For horizontal wires we use the same convention but rotated 90 degrees counterclockwise, i.e.\ if the $1$s are on the west half of the unmarked bits, we carry a logical $0$-bit.
In Figure~\ref{fig:wire}, the leftmost charged preimage thus encodes a $1$, and the middle one encodes a $0$.

\section{The basic gadgets}
\label{sec:BasicGadgets}


In this section, we introduce the so-called \emph{basic gadgets} which ``manipulate'' signals traversing a wire. Again, a wire is \emph{charged} if (in a particular preimage or family of preimages) its preimage consists of the $W_i$-patterns. As mentioned in the introduction, we verify such properties of our gadgets by SAT solvers, and we have not checked (all of) them through other means; thus, this section simply \emph{describes} the properties of the gadgets, and does not prove them. See Section~\ref{sec:Implementation} for details on how we found and verified the gadgets.

Formally, a rectangle-shaped pattern can be thought of as defining a relation on the cells on the boundary of its preimages. All of our basic gadgets are rectangular patterns satisfying the following common properties, which will allow us to easily combine them into computational gadgets.
Some cells on the boundary of the gadget are part of wires pointing out of the domain. Each of the four borders will have at most one wire.
The gadget defines some relation on the possible phases of signals on the wires, by admitting some combinations in a preimage and forbidding others.
Each gadget carries the information about the operating phases of each of its boundary wires.
Each realizable combination of operating phases on the boundary wires can also be realized with zeroes everywhere else on the boundary.

The last item is a key property, as it allows us to worry only about the wires: the relations on the wires already force whatever computation we want to happen, and for any such computation, we know there is a legal preimage, because we can simply write zeroes everywhere on the boundaries and stitch the gadgets together along them.

We now define a notation for the precise behavior of our gadgets.
The first property we are interested in is its \emph{charging behavior}: which of the boundary wires are required to be charged in an operating phase for the gadget to function properly, and which ones are charged by the gadget itself?
We denote this property by $X \vdash Y$, where $X$ lists those directions among $\{E,S,W,N\}$ that should be charged in an operating phase from the outside, and $Y$ lists those that are automatically charged in an operating phase assuming the wires in $X$ are.
We might give several lists for $X$, with the interpretation that any one choice among them will charge all the wires in $Y$.
For example, the charging behavior $EW, S \vdash EN$ means that in any particular preimage, if the east and west wires are charged, or the south wire is charged, then the east and north wires will also be charged (in their respective operating phases).

The second property is the \emph{relation} between the binary values carried by the boundary wires: which combinations of values are realizable in a preimage?
We denote this property by simply listing all the realizable combinations as binary words.
For example, the relation $EWN \in \{000, 001, 010, 100\}$ means that there are wires on the east, west and north sides of the gadget, and as long as all three are charged in their operating phases, at most one of them can carry a 1-signal.
We remind the reader that by construction of the gadgets, each realizable combination will also be realizable by a preimage with $0$s on the border except near the wires.
The charging behavior and relation of a gadget together constitute its \emph{specs}.




\begin{example}
\label{ex:Specs}
  Consider a hypothetical gadget with the specs
  \[
    W \vdash NE; WNE \in \{001, 010, 100\}.
  \]
  This means that the gadget has wires on the west, north and east sides.
  If the west wire is charged, meaning that a preimage contains a signal in an operating phase, then the two others will be charged as well (in that preimage).
  In any such preimage, exactly one of the wires is in phase $1$, and all such combinations can be realized in preimages where the rest of the boundary consists of $0$s. \qee
\end{example}

We note that the specs refer to a particular orientation of the gate. When a gate is rotated 90 degrees counterclockwise, of course the permutation $(E \, \; N \, \; W \, \; S)$ is applied to all symbols appearing in the specification. In addition, when a wire turns counterclockwise from west to south or east to north, one must flip the bit on that wire in every tuple in the relation. Thus, in the relation, $ENWS = abcd$ turns into $ENWS = d \bar{a} b \bar{c}$, where $\bar 0 = 1, \bar 1 = 0$.

This convention for the meaning of $0$ and $1$ in a wire makes it easy to reason about large circuits, as we need not keep track of an orientation for individual wires. However, when reasoning about individual gadgets (and especially when rotating them) an easier interpretation is that of orientations. A horizontal wire coming into a gadget from the west (resp.\ east) boundary is said to have \emph{near} charge if its charge carries the bit $1$ (resp.\ $0$). Similarly, a vertical wire coming into a gadget from the north (resp.\ south) boundary is said to have \emph{near} charge if its charge carries the bit $1$ (resp.\ $0$). Intuitively, looking at the two operating phases of the signal between any two marked cells, the one closer to the center of the gadget is the near charge. The other operating phase is called a \emph{far} charge. This determines the \emph{affinity} of the charge relative to the gadget.
See Figure~\ref{fig:affinity} for an illustration.

\begin{figure}[htp]
  \begin{center}
    \begin{tikzpicture}[scale = 0.25]
      \pgfplotstableread{wire.cvs}{\matrixfile}
      \kek
      \draw[black!30!white] (0,-1) grid (12,1);
      \node[left] at (-1,0) {Gadget 1};
      \node[right] at (13,0) {Gadget 2};
      \draw (1,-1) -- (1,-3) -- (-9,-3) -- (-9,3) -- (1,3) -- (1,1);
      \draw (11,-1) -- (11,-3) -- (21,-3) -- (21,3) -- (11,3) -- (11,1);

      \begin{scope}[yshift=-6cm]
        \pgfplotstableread{wirepre1.cvs}{\matrixfile}
        \kek
        \draw[black!30!white] (0,-3) grid (12,1);
        \node[left] at (-1,-1) {1-bit: near charge};
        \node[right] at (13,-1) {far charge};
      \end{scope}

      \begin{scope}[yshift=-12cm]
        \pgfplotstableread{wirepre2.cvs}{\matrixfile}
        \kek
        \draw[black!30!white] (0,-3) grid (12,1);
        \node[left] at (-1,-1) {0-bit: far charge};
        \node[right] at (13,-1) {near charge};
      \end{scope}
    \end{tikzpicture}
  \end{center}
  \caption{The affinity of a preimage of a charged wire near the boundary of a gadget.}
  \label{fig:affinity}
\end{figure}

If gadgets are thought of as manipulating the affinity of charges, the specs stay the same when we rotate them. 
Of course the conversion between the two notations is straightforward.
For the reader's comfort we include in all specs the affinity version of the relation, as a subset of $\{\near,\far\}^k$ for a suitable power $k$, with $\near$ and $\far$ denoting respectively the near and far charges.
For example, the gadget of Example~\ref{ex:Specs} has the relation $WNE \in \{\far\far\far, \far\near\near, \near\far\near\}$.


\subsection{Charging a wire: the charger-splitter combo}

A typical preimage of a wire is messy, and does not propagate deterministically, so to get computation going the first thing we need to do is charge our wires.
In other words, we want a gadget such that in every preimage, the wire on the boundary is charged, and there is a preimage where the wire is in this charged state and everything else on the boundary is zero. The \emph{charger gadget}, shown in Figure~\ref{fig:charger}, has such a behavior.

\begin{figure}[htp]
  \pgfplotstableread{charger.cvs}{\matrixfile}
  \begin{center}
    \begin{tikzpicture}[scale = 0.25]
\kek
      \draw[black!30!white] (0,-18) grid (24,1);
      \draw[thick,black] (0,-18) rectangle (24,-2);
      \draw[thick,gray] (12,-2) ellipse (2.5 and 1.25);
    \end{tikzpicture}
  \end{center}
  \caption{The charger gadget where one output is near and one is far. Specs: $\vdash N; N \in \{0, 1\}/\{\near, \far\}$. The gadget is inside the black rectangle, and we show the beginning of a wire on top.}
  \label{fig:charger}
\end{figure}


In words, the precise property of the gadget is the following.
The $4 \times 2$ rectangle at the top boundary, whose bottom half is inside the domain and which has the two black top cells of the wire in its middle (circled in Figure~\ref{fig:charger}), is forced to contain the pattern $W_0$ or $W_1$.
For each of the two, the gadget has a preimage where the boundary is otherwise empty. Intuitively, the behavior of the charger gadget is quite simple: it charges the top wire in an arbitrary operating phase.


A charged wire that starts from a charger is somewhat useless on its own.
In order to transmit information, we need charged wires with two ends. To achieve this, we combine the charger with the \emph{splitter}, depicted in Figure~\ref{fig:splitter}, which splits a signal in two directions.
In fact, we will only use the charger in combination with a splitter.

Note that there is a wire going straight through the splitter, so the north-south behavior is clear, and the wire on the east copies this information. In the shown orientation, the gadget has the side effect that the signal moving to the east is inverted. In terms of affinity, we copy affinity from the side where a bar of four black cells -- the \emph{handle bar} -- sticks out of the wire (on top in the shown orientation).

\begin{figure}[htp]
  \pgfplotstableread{splitter.cvs}{\matrixfile}
  \begin{center}
    \begin{tikzpicture}[scale = 0.25]

     \kek
      \draw[black!30!white] (0,-19) grid (14,1);
      \draw[thick,black] (0,-19) rectangle (14,1);
      \node (handle) at (11,3) {handle bar};
      \draw (handle) edge[-stealth] (8,0.5);
      \draw[gray,thick] (6,-0.5) ellipse (2.5 and 1);
    \end{tikzpicture}
  \end{center}
  \caption{The splitter. Specs: $N, S \vdash ENS; ENS \in \{100, 011\}/\{\far \far \near, \near \near \far\}$.} 
  \label{fig:splitter}
\end{figure}

Combining the charger and splitter, we obtain the \emph{charged turn} gadgets shown in Figure~\ref{fig:ers}. The intuitive behavior of these gadgets is that the input and output wires are charged and their signals are synchronized. Depending on the orientation, the signal may be flipped. In terms of affinity, if the signal meets the handle bar, affinity is kept.

\begin{figure}[htp]
  \begin{center}
    \pgfplotstableread{CScombo1.cvs}{\matrixfile}
    \begin{tikzpicture}[scale = 0.22]
	 \kek
      \draw[black!30!white] (-5,-34) grid (19,3);
      \draw[thick,black] (-5,-34) rectangle (19,3);
    \end{tikzpicture}
     \pgfplotstableread{CScombo2.cvs}{\matrixfile}
    \begin{tikzpicture}[scale = 0.22]
	 \kek
      \draw[black!30!white] (-5,-34) grid (19,3);
      \draw[thick,black] (-5,-34) rectangle (19,3);
    \end{tikzpicture}
  \end{center}
  \caption{The charged turn gadgets. Specs: $\vdash ES$ for both, $ES \in \{01, 10\}/\{\near \far, \far \near\}$ for the first, $ES \in \{00, 11\}/\{\near \near, \far \far\}$ for the second.}
  \label{fig:ers}
\end{figure}

\subsection{Shifting and inverting a signal}

The charged turn allows us to construct charged wires and turn them. However, for now the wires are cumbersome, if not impossible, to work with: our gadgets do not align nicely modulo three (the period of the signals), and we cannot yet shift the phase of the signal. Thus, we introduce a gadget that allows the free manipulation of the phase, namely the \emph{inverter}, shown in Figure~\ref{fig:phaseshifternot}.

\begin{figure}[htp]
   \pgfplotstableread{not.cvs}{\matrixfile}
   \begin{center}
    \begin{tikzpicture}[scale = 0.25]
\kek
      \draw[black!30!white] (0,-17) grid (28,1);
      \draw[thick,black] (0,-17) rectangle (28,1);
      \node (handle) at (7.75,-15.5) {kicker};
      \draw (handle) edge[-stealth] (9.25,-13.5);
      \draw[thick,gray] (10,-12) ellipse (1.5 and 1.5);
    \end{tikzpicture}
  \end{center}
  \caption{The inverter. Specs: $E \vdash W; EW \in \{01, 10\}/\{\near \near, \far \far\}$.}
  \label{fig:phaseshifternot}
\end{figure}

As a mnemonic, the reader may find it useful to visualize the contents of the bottom two inhabited rows as the ``feet'' of the gadget. In terms of these feet, the inverter kicks the signal with its smaller foot, or its \emph{kicker} (in the figure, the westmost one) to invert it. As a side effect it 
increments the phase of the signal by one, i.e.\ the horizontal distance between a marked cell on the west side and one on the east side is $-1 \bmod 3$. The inverter gate propagates a charge in the direction of the kick, which in the figure is from east to west.

Since $2$ and $3$ are coprime, the inverter on its own can be used to invert and phase shift the signal. Namely, we can first use an inverter to change the bit carried by a wire if desired, and then two copies of the inverter to change the phase by $1 \bmod 3$ or four copies to change it by $-1 \bmod 3$. It seems likely that gadgets of similar size exist for all phase-signal transformations, but we have not been able to find any.

\subsection{Wire crossing}

A \emph{wire-crossing} gadget is given in Figure~\ref{fig:crossing}. Note that the gadget performs no charging, and that as a side-effect, phases are inverted (i.e.\ affinity is copied).

\begin{figure}[htp]
  \pgfplotstableread{crossing.cvs}{\matrixfile}
  \begin{center}
    \begin{tikzpicture}[scale = 0.25]
\kek
      \draw[black!30!white] (0,-29) grid (30,1);
      \draw[thick,black] (0,-29) rectangle (30,1);
    \end{tikzpicture}
  \end{center}
  \caption{The wire-crossing gadget. Specs: no charging, $ENWS \in \{0011, 0110, 1001, 1100\}/\{\near \far \near \far, \near \near \near \near, \far \far \far \far, \far \near \far \near\}$.}
  \label{fig:crossing}
\end{figure}

\subsection{A universal logical gate}
\label{sec:logic-gate}

Now that we can freely connect wires and invert their values, all we need is a gate that, together with NOT, is universal. One such gate is specified in Figure~\ref{fig:and}. In terms of affinity, the top output is far if and only if both inputs are far. In global terms, in the shown orientation, seeing the west and east signal phases as the input, and north as the output, this is the $N = W \vee \neg E$ gate; the north signal is $0$ if and only if the west signal is $0$ and the east signal is $1$. By composing with inverters, one can easily produce any standard gate such as an AND gate or an OR gate in any orientation.

\begin{figure}[htp]
  \pgfplotstableread{and.cvs}{\matrixfile}
  \begin{center}
    \begin{tikzpicture}[scale = 0.25]
\kek
      \draw[black!30!white] (0,-33) grid (25,1);
      \draw[thick,black] (0,-33) rectangle (25,1);
    \end{tikzpicture}
  \end{center}
  \caption{A gadget that performs the logical operation $N = W \vee \neg E$. Specs: no charging, $ENW \in \{111, 011, 100, 010\}/\{\far \near \far, \near \near \near, \far \far \far, \near \near \far\}$.}
  \label{fig:and}
\end{figure}

\subsection{Constant charger}

The above set of gadgets is already universal, but we nevertheless present also a gadget that introduces a fixed charge on a wire. It is called the \emph{enforcer} and is depicted in Figure~\ref{fig:enforcer}. As the enforcer sends a fixed charge, it does not determine two phases, so we can arbitrarily choose whether the signal it carries is $0$ or $1$ by aligning it appropriately. We choose to interpret the signal as being sent in the near phase. In the figure we use the usual white dot to mark the cells where the preimage always contains zero, i.e.\ the unused phase, and the white cross denotes the other unused phase coming from the fact the $\far$ bit is never emitted.

\begin{figure}[htp]
   \pgfplotstableread{enforcer.cvs}{\matrixfile}
   \begin{center}
    \begin{tikzpicture}[scale = 0.25,rotate=90]
\kek
      \draw[black!30!white] (0,-31) grid (20,1);
      \draw[thick,black] (0,-31) rectangle (20,-6);
    \end{tikzpicture}
  \end{center}
  \caption{The enforcer. Specs: $\vdash W; W \in \{1\}/\{\near\}$.}
  \label{fig:enforcer}
\end{figure}




\section{Composite gadgets and abstract circuits}
\label{sec:CompositeGadgets}

Our composite gadgets will correspond to the following set of \emph{gate tiles}:
\[ T = \{\gateE, \gateH, \gateV, \gateEN, \gateNW, \gateWS, \gateES, \gateN, \gateT, \gateS, \gateX, \gateO\}. \]
Before describing the gadgets, we give these symbols abstract component semantics, which mostly correspond to their usual meanings as used in digital logic. A \emph{well-formed circuit} is any element of the SFT $X_T \subset T^{\Z^2}$ where in adjacent tiles all wires that reach the boundary of the tile must continue in the neighboring tile.
If $T$ is interpreted as a set of Wang tiles over the color set $C = \{\text{no wire}, \text{wire}\}$ (ignoring the fact that some combinations of colors correspond to several distinct tiles), then $X_T$ is precisely the SFT defined by them as Wang tiles.

A circuit is \emph{satisfiable} if it is possible to pick a \emph{signal} $s \in \{0,1\}$ for each of the wires in each tile reaching its boundary, so that the following hold:
\begin{itemize}
\item In the tiles $\gateH, \gateV, \gateEN, \gateNW, \gateWS, \gateES$ the signals at both wires are the same.
\item In $\gateN$, the signals in the two wires are distinct.
\item In $\gateT$, the signal in the unique wire is $1$.
\item In $\gateS$, the signal in all the wires is the same.
\item In $\gateX$, the north and south signals are equal, and the east and west signals are equal.
\item In $\gateO$, the east signal is $1$ if and only if at least one of the north and south signals is $1$. (This is the logical OR operation.)
\end{itemize}
Write $Z_T \subset X_T \times ((\{0,1,\bot\})^4)^{\Z^2}$ for the subshift encoding the satisfying assignments (where the $\{0,1,\bot\})^4$-element in each cell codes the values of signals, and $\bot$ is used iff there is no wire), and $f : Z_T \to T^{\Z^2}$ the natural projection. (For technical reasons, which will become clear later, the codomain is be taken to be a full shift.) We denote $f(Z_T) = Y_T \subset X_T$. The \emph{circuit system} is the triple $(f, Z_T, T^{\Z^2})$.

Note that we are ``missing'' some obvious gates, like an AND gate and various orientations of gates, which would be useful for many engineering purposes, but they can be ``simulated'' easily (even in a precise sense as discussed in Section~\ref{sec:Universality}). 


We now explain what our composite gadgets should do, and the correspondence with gate tiles. To each gate tile, we will associate a $450 \times 450$ pattern with four fixed positions on the east, north, west and south boundaries, called its \emph{connector positions}. The east and west connector positions have the same y-coordinate, and the north and south connector positions have the same x-coordinate, so that the connector positions of adjacent gate tiles will match. A wire enters the gadget from a connector position if and only if the corresponding gate tile has an outgoing wire on that side. Thus the SFT rule of $X_T$ translates into ``wires should continue across the borders of adjacent composite gadgets''.

The rest of the boundary of each composite gadget is all zero. In every $g$-preimage of a composite gadget, all the boundary wires are charged in one of two phases. The two phases used are the ones where the $1$s are on either side of the boundary (the phase where we have zeroes around the boundary never extends to a preimage of a gadget). We refer to the choice of which two phases are used for bits as the \emph{phase alignment}, and the aforementioned choice of positioning of the used phases is referred to as the \emph{neutral alignment}. 

Interpreting the signal values as in the previous section, i.e.\ a representing a $1$-signal by having the $1$s on the east and south sides of the boundary, the allowed combinations for the signals should be precisely the ones described above for the abstract semantics, and furthermore for any such choice the rest of the boundary of each composite gadget should be fillable with zeroes.

We construct such gadgets in two parts. First, we give $90 \times 90$ blocks that simply connect the corresponding gadget from Section~\ref{sec:BasicGadgets} to a consistent position on each border (we position our wire boundaries at the $29$th and $30$th column/row), in some cases using some inverters. We cannot simply put the resulting $90 \times 90$ patterns together, since the wires are not necessarily charged, and the phase alignments of neighboring patterns may not match.

To solve these issues, we introduce horizontal (resp.\ vertical) ``charged wire gadgets'' of shape $180 \times 90$ (resp. $90 \times 180$). These gadgets should charge the wires, and we construct one for each of the nine possible alignment changes. We can then position our basic gadgets in the middle $90 \times 90$ squares of $450 \times 450$ squares, and use the remainder of the space to put charged wire gadgets on the appropriate sides, changing all phase alignments to the neutral one.

Figure~\ref{fig:450} shows the $450 \times 450$ gadget corresponding to the OR gate.
The central $90 \times 90$ part is constructed from the logical cate of Section~\ref{sec:logic-gate} by adding inverters to two of its wires, so that its semantics match the standard OR gate.
Each of the three charged wire gadgets consists of four charged turns that charge the wire and some number of inverters that correct its phase to the neutral one.

\newcommand{\picwd}{7}
\begin{figure}[htp]
  \begin{center}
    \begin{tikzpicture}
      \node at (0,0) {\includegraphics[width=\picwd cm,height=\picwd cm]{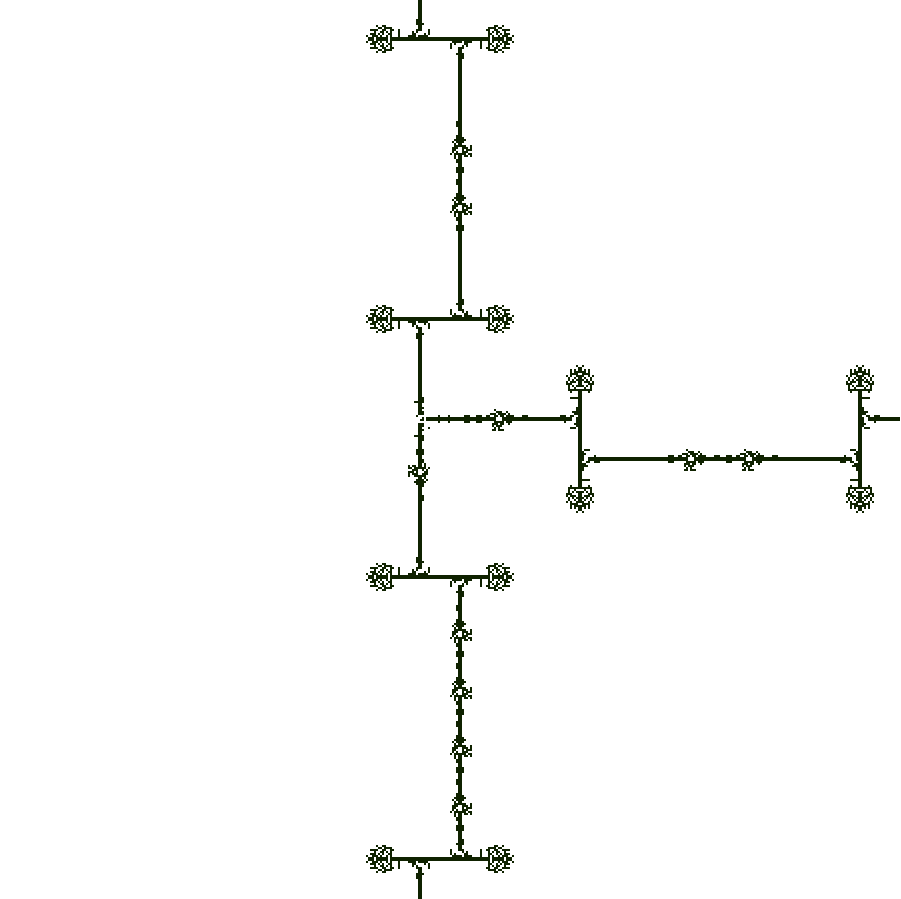}};
      \draw (-\picwd/2,-\picwd/2) rectangle (\picwd/2,\picwd/2);
      \draw [dashed] (-\picwd/2,-\picwd/10) -- (\picwd/2,-\picwd/10);
      \draw [dashed] (-\picwd/2,\picwd/10) -- (\picwd/2,\picwd/10);
      \draw [dashed] (-\picwd/10,-\picwd/2) -- (-\picwd/10,\picwd/2);
      \draw [dashed] (\picwd/10,-\picwd/2) -- (\picwd/10,\picwd/2);
    \end{tikzpicture}
  \end{center}
  \caption{A $450 \times 450$ composite gadget for the OR gate.}
  \label{fig:450}
\end{figure}

We provide Golly-compatible patterns corresponding to this description in the repository \cite{pat-repo}. If one has memorized the specs of the basic gadgets given in Section~\ref{sec:BasicGadgets}, it is a reasonably simple eyeballing exercise to verify that they indeed correctly implement the semantics of the gate tiles. 
In addition to having performed this eyeballing exercise several times ourselves, we have checked the semantics of the composite gadgets (but not the $450 \times 450$ squares) by a SAT solver.
A Python script that performs this check for all the gadgets of Section~\ref{sec:BasicGadgets}, the corresponding $90 \times 90$ squares, and the $180 \times 90$ charged wires, can also be found in \cite{pat-repo}; its running time is only a few seconds.
The SAT instances that verify the gadgets are not difficult to generate programmatically, so a skeptical reader can also write their own verification script.

Once we have the gadgets, it is easy to see that for any satisfiable circuit, the corresponding tiling by composite gadgets has a preimage: satisfy the abstract circuit, copy the signal values (using the interpretation of values which is consistent among neighboring composite gadgets), and fill the rest of the boundaries with zeroes. On the other hand, if the composite gadget has a preimage, then already the values on the wires prove the satisfiability of the circuit. Thus, we obtain the following theorem -- which is, again, valid assuming the verification script is correct.

\begin{theorem}
\label{thm:Sub}
There exists a substitution $\tau : T \to \{0,1\}^{450 \times 450}$ such that for $x \in X_T$, $\tau(x)$ has a $g$-preimage if and only if $x \in Y_T$.
\end{theorem}

\section{A concrete example: the Jeandel-Rao tile set}
\label{sec:JeandelRao}

Before delving deeper into universality, we show how to simulate an arbitrary Wang tile set with at most four edge colors in a direct fashion, without explicitly referring to universality. We also optimize the size of the simulation slightly, from $450 \times 450$ to $270 \times 270$, by not resetting to neutral phase alignment between patterns, but simply directly using, between two composite gadgets, the correct charged wire gadget connecting their phase alignments.

The generalization of the construction of this section to an arbitrary tile set is straightforward, and can be used to give an alternative proof of Lemma~\ref{thm:CircuitsStronglyUniversal}, which for small tile sets gives smaller and simpler implementations.

Represent the colors of a Wang tile by eight bits $x_1, x_2, \ldots, x_8$, with the convention that $(x_2, x_1)$ are the west color, $(x_3, x_4)$ the north color, $(x_5, x_6)$ is the east color, and $(x_8, x_7)$ is the south color. A set of $k$ Wang tiles then becomes a set of $k$ words in $\{0,1\}^8$, and we can represent it as a logical formula in disjunctive normal form with $k$ clauses.

\begin{figure}[htp]
  \pgfplotstableread{JeandelRao.cvs}{\matrixfile}
  \begin{center}
    \begin{tikzpicture}[scale = 0.4]

      \pgfplotstablegetrowsof{\matrixfile} 
      \pgfmathtruncatemacro{\totrow}{\pgfplotsretval}
      \pgfplotstablegetcolsof{\matrixfile} 
      \pgfmathtruncatemacro{\totcol}{\pgfplotsretval}
      \draw[black!30!white] (0,-43) grid (23,1);
      \pgfplotstableforeachcolumn\matrixfile\as\col{
        \pgfplotstableforeachcolumnelement{\col}\of\matrixfile\as\colcnt{%
          \ifnum\colcnt=1
		 \draw ($ -\pgfplotstablerow*(0,\cellht) + \col*(\cellwd,0) +(0,0.5)$) -- ($ -\pgfplotstablerow*(0,\cellht) + \col*(\cellwd,0) + (1,0.5)$);
		 \fi
		 \ifnum\colcnt=2
		 \draw ($ -\pgfplotstablerow*(0,\cellht) + \col*(\cellwd,0) +(0.5,0)$) -- ($ -\pgfplotstablerow*(0,\cellht) + \col*(\cellwd,0) + (0.5,1)$);
		 \fi
		 \ifnum\colcnt=3 
		 \draw ($ -\pgfplotstablerow*(0,\cellht) + \col*(\cellwd,0) +(1,0.5)$) -- ($ -\pgfplotstablerow*(0,\cellht) + \col*(\cellwd,0) + (0.5,0.5)$);
		 \draw ($ -\pgfplotstablerow*(0,\cellht) + \col*(\cellwd,0) +(0.5,0.5)$) -- ($ -\pgfplotstablerow*(0,\cellht) + \col*(\cellwd,0) + (0.5,1)$);
		 \fi
		 \ifnum\colcnt=4 
		 \draw ($ -\pgfplotstablerow*(0,\cellht) + \col*(\cellwd,0) +(0,0.5)$) -- ($ -\pgfplotstablerow*(0,\cellht) + \col*(\cellwd,0) + (0.5,0.5)$);
		 \draw ($ -\pgfplotstablerow*(0,\cellht) + \col*(\cellwd,0) +(0.5,0.5)$) -- ($ -\pgfplotstablerow*(0,\cellht) + \col*(\cellwd,0) + (0.5,1)$);
		 \fi'
		 \ifnum\colcnt=5 
		 \draw ($ -\pgfplotstablerow*(0,\cellht) + \col*(\cellwd,0) +(0,0.5)$) -- ($ -\pgfplotstablerow*(0,\cellht) + \col*(\cellwd,0) + (0.5,0.5)$);
		 \draw ($ -\pgfplotstablerow*(0,\cellht) + \col*(\cellwd,0) +(0.5,0.5)$) -- ($ -\pgfplotstablerow*(0,\cellht) + \col*(\cellwd,0) + (0.5,0)$);
		 \fi
		 \ifnum\colcnt=6 
		 \draw ($ -\pgfplotstablerow*(0,\cellht) + \col*(\cellwd,0) +(1,0.5)$) -- ($ -\pgfplotstablerow*(0,\cellht) + \col*(\cellwd,0) + (0.5,0.5)$);
		 \draw ($ -\pgfplotstablerow*(0,\cellht) + \col*(\cellwd,0) +(0.5,0.5)$) -- ($ -\pgfplotstablerow*(0,\cellht) + \col*(\cellwd,0) + (0.5,0)$);
		 \fi
		 \ifnum\colcnt=7 
		 \draw ($ -\pgfplotstablerow*(0,\cellht) + \col*(\cellwd,0) +(0,0.5)$) -- ($ -\pgfplotstablerow*(0,\cellht) + \col*(\cellwd,0) + (1,0.5)$);
		 \filldraw[black] ($ -\pgfplotstablerow*(0,\cellht) + \col*(\cellwd,0) +(0.5,0.5)$) circle (0.15);
		 \fi
		 \ifnum\colcnt=8 
		 \draw ($ -\pgfplotstablerow*(0,\cellht) + \col*(\cellwd,0) +(0,0.5)$) -- ($ -\pgfplotstablerow*(0,\cellht) + \col*(\cellwd,0) + (0.5,0.5)$);
		 \draw ($ -\pgfplotstablerow*(0,\cellht) + \col*(\cellwd,0) +(0.5,0.5)$) -- ($ -\pgfplotstablerow*(0,\cellht) + \col*(\cellwd,0) + (0.7,0.7)$);
		 \draw ($ -\pgfplotstablerow*(0,\cellht) + \col*(\cellwd,0) +(0.5,0.5)$) -- ($ -\pgfplotstablerow*(0,\cellht) + \col*(\cellwd,0) + (0.7,0.3)$);
		 \fi
		 \ifnum\colcnt=9 
		 \draw ($ -\pgfplotstablerow*(0,\cellht) + \col*(\cellwd,0) +(0.5,0.5)$) -- ($ -\pgfplotstablerow*(0,\cellht) + \col*(\cellwd,0) + (1,0.5)$);
		 \draw ($ -\pgfplotstablerow*(0,\cellht) + \col*(\cellwd,0) +(0.5,0)$) -- ($ -\pgfplotstablerow*(0,\cellht) + \col*(\cellwd,0) + (0.5,1)$);
		 \fi
		 \ifnum\colcnt=10 
		 \draw ($ -\pgfplotstablerow*(0,\cellht) + \col*(\cellwd,0) +(0,0.5)$) -- ($ -\pgfplotstablerow*(0,\cellht) + \col*(\cellwd,0) + (1,0.5)$);
		 \draw ($ -\pgfplotstablerow*(0,\cellht) + \col*(\cellwd,0) +(0.5,0)$) -- ($ -\pgfplotstablerow*(0,\cellht) + \col*(\cellwd,0) + (0.5,1)$);
		 \fi
		 \ifnum\colcnt=11 
		 \draw ($ -\pgfplotstablerow*(0,\cellht) + \col*(\cellwd,0) +(0.5,1)$) -- ($ -\pgfplotstablerow*(0,\cellht) + \col*(\cellwd,0) + (0.5,0.75)$);
		 \draw ($ -\pgfplotstablerow*(0,\cellht) + \col*(\cellwd,0) +(0.5,0)$) -- ($ -\pgfplotstablerow*(0,\cellht) + \col*(\cellwd,0) + (0.5,0.25)$);
		 \draw ($ -\pgfplotstablerow*(0,\cellht) + \col*(\cellwd,0) +(0.8,0.5)$) -- ($ -\pgfplotstablerow*(0,\cellht) + \col*(\cellwd,0) + (1,0.5)$);
		 \draw ($ -\pgfplotstablerow*(0,\cellht) + \col*(\cellwd,0) +(0.25,0.2)$) edge[bend right=80,in=-90, out=-90, looseness = 1] ($ -\pgfplotstablerow*(0,\cellht) + \col*(\cellwd,0) + (0.25,0.8)$);
		 \draw ($ -\pgfplotstablerow*(0,\cellht) + \col*(\cellwd,0) +(0.25,0.2)$) edge[bend right=80,in=-90, out=-90, looseness = 3] ($ -\pgfplotstablerow*(0,\cellht) + \col*(\cellwd,0) + (0.25,0.8)$);
		 \fi
		 \ifnum\colcnt=12 
		 \draw ($ -\pgfplotstablerow*(0,\cellht) + \col*(\cellwd,0) +(0,0.5)$) -- ($ -\pgfplotstablerow*(0,\cellht) + \col*(\cellwd,0) + (0.2,0.5)$);
		 \draw ($ -\pgfplotstablerow*(0,\cellht) + \col*(\cellwd,0) +(0.8,0.5)$) -- ($ -\pgfplotstablerow*(0,\cellht) + \col*(\cellwd,0) + (1,0.5)$);
		 \node () at ($ -\pgfplotstablerow*(0,\cellht) + \col*(\cellwd,0) +(0.5,0.5)$) {\footnotesize $?$};
		 \fi
       }        
      }
     
      \draw[thick,black] (0,-43) rectangle (23,1);
      
      \node[anchor=east] (rows) at (-1,-22) {rows $17$--$28$};
      \draw [decorate, decoration = {brace}] (-0.5,-28) --  (-0.5,-16);
      
      \node () at (-0.7,-3.5) {$x_1$};
      \node () at (-0.7,-2.5) {$x_2$};
      
      \node () at (4.5,1.7) {$x_3$};
      \node () at (3.5,1.7) {$x_4$};

      \node () at (23.7,-3.5) {$x_6$};
      \node () at (23.7,-2.5) {$x_5$};
      
      \node () at (4.5,-43.7) {$x_8$};
      \node () at (3.5,-43.7) {$x_7$};
      
    \end{tikzpicture}
  \end{center}
  \caption{The ``blueprint'' for simulating a tile set by preimages of $g$.}
  \label{fig:Blueprint}
\end{figure}

Consider Figure~\ref{fig:Blueprint}. Repeat the rows 17--28 (0-indexed) $k-2$ times. More precisely, seeing the figure as a matrix, add $12(k-2)$ more rows immediately after the $28$th, copying the contents periodically from rows $17$--$28$, and leave the first $17$ and last $15$ rows intact.

Next, observe that, picking the ?s to be either horizontal wires or NOT-gates, we can interpret the binary tree of $7$ OR-gates followed by a NOT gate as a conjunctive clause (using De Morgan's law, i.e.\ inverting the inputs). We can then interpret the column of ORs on the right as a disjunction of clauses.

Finally, we replace the gate tiles by the $90$-by-$90$ blocks from the repository \cite{pat-repo}, and connect them with charged wire gadgets -- the minor additional complication is now that we need to pick the phases context-sensitively, i.e.\ we have to make sure that the outgoing phase of a wire going out from one gadget is matched with the incoming phase of the neighboring one. This can be done, as the repository lists the gadgets for connecting any pair of phases.

\begin{theorem}
\label{thm:AperiodicOnly}
  There exists a totally periodic configuration $x \in \{0,1\}^{\Z^2}$ with periods $6210 \times 37800$ that has a $g$-preimage, but no periodic preimage.
\end{theorem}

\begin{proof}
The aperiodic Jeandel-Rao tile set \cite{JeRa21} uses only four colors on each side and has $11$ tiles. With $k = 11$, after the rows $19$--$31$ have been repeated $k-2$ times we have an array with $17 + 12 \cdot 9 + 15 = 140$ rows and $23$ columns. Replacing each cell by the appropriate $270 \times 270$ pattern, we get the claimed periods of $6210 \times 37800$.
\end{proof}

The fundamental domain of the configuration $x$ is shown in Figure~\ref{fig:JeandelRao}.

\section{Universality}
\label{sec:Universality}

In this section, we give more precise universality definitions, and state a stronger version of Theorem~\ref{thm:Sub}. Our precise definition of simulation is a bit involved, and readers without a background or an interest in universality per se may be better off skipping straight to Section~\ref{sec:Corollaries}.
Indeed, the proof of Theorem~\ref{thm:Sub} and the construction from Section~\ref{sec:JeandelRao} can more or less replace the universality discussion in terms of understanding Section~\ref{sec:Corollaries} and proving the results therein.




\begin{definition}
  Let $X \subset A^{\Z^2}$ be a subshift.
  For $m, n \in \N$, the \emph{$(m, n)$-blocking} of $X$ is the $\Z^2$-dynamical system $(X, m\Z \times n\Z)$, where $\Z^2$ acts as $(i,j) \cdot x = \sigma^{(mi, nj)}(x)$.
  It can be seen as a subshift over the alphabet of blocks $A^{m \times n}$ in a natural way.
  Every block map $\psi : X \to Y$ lifts into a block map between blockings $\psi^{m \times n} : (X, m\Z \times n\Z) \to (Y, m\Z \times n\Z)$, called the $(m, n)$-blocking of $\psi$.
  We denote it simply by $\psi$ when there is no danger of confusion.
\end{definition}

We interpret the $(m,n)$-blocking of a subshift $X$ as a set of configurations with a superimposed $m \times n$ grid.
A block map $f : (X, m\Z \times n\Z) \to Y$ needs to satisfy the relation $f(\sigma^{(mi, nj)}(x)) = \sigma^{(i,j)}(f(x))$ for all $(i,j) \in \Z^2$ and $x \in X$, that is, a shift by $(1,0)$ on the codomain side corresponds to a shift by $(m,0)$ on the domain side, and similarly for $(0,1)$ and $(0,n)$.
Similarly, a substitution $\tau : A \to B^{m \times n}$ of shape $m \times n$ can be interpreted as a block map $\tau : A^{\Z^2} \to (B^{\Z^2}, m\Z \times n\Z)$.

In the definition of simulations, we discuss commutation of diagrams of \emph{partial functions}.
To make sense of this, we make explicit our formalism of partial functions.

\begin{definition}
  A \emph{partial function} $f : A \nrightarrow B$ consists of a domain $\mathrm{dom}(f) \subset A$ and an underlying function $\mathrm{fun}(f) : \mathrm{dom}(f) \to B$, which we may abusively denote by $f$.
  Two partial functions are equal if they have equal domains and underlying functions.
  We say that $f$ is injective or surjective if $\mathrm{fun}(f)$ is, and bijective if $\mathrm{dom}(f) = A$ and $\mathrm{fun}(f)$ is bijective.
  The composition $g \circ f : A \nrightarrow C$ of two partial functions $f : A \nrightarrow B$ and $g : B \nrightarrow C$ is the partial function with domain $\mathrm{dom}(f) \cap \mathrm{fun}(f)^{-1}(\mathrm{dom}(g))$ and the obvious underlying function.
\end{definition}

\begin{definition}
  Let $\psi : X \to B^{\Z^2}$ be a block map, where $X \subset A^{\Z^2}$, and let $C \subset B$.
  The \emph{alphabet co-restriction to $C$} of $\psi$ is the partial block map $\psi_C : X \nrightarrow B^{\Z^2}$ with domain $\psi^{-1}(C^{\Z^2})$ and underlying function the restriction of $\psi$.
\end{definition}



\begin{definition}
\label{def:Simulations}
Let $\psi : X \to B^{\Z^2}, \phi : Y \to D^{\Z^2}$ be block maps, where $X \subset A^{\Z^2}, Y \subset C^{\Z^2}$ are subshifts of finite type. We say that $\psi$ \emph{weakly simulates} $\phi$ if there exist $m, n \in \N$, an injective substitution $\tau : D \to B^{m \times n}$, and a surjective partial block map $h : (X, m\Z \times n\Z) \nrightarrow Y$, such that the following diagram of partial functions commutes:
\begin{equation}
  \label{eq:Diagram}
  \begin{tikzcd}
    (X, m\Z \times n\Z) \arrow[negate]{r}{h} \arrow[swap,negate]{d}{\psi_{\tau(D)}} & Y \arrow{d}{\phi} \\
    (B^{\Z^2}, m\Z \times n\Z) & D^{\Z^2} \arrow{l}{\tau}
  \end{tikzcd}
\end{equation}
We say $\psi$ \emph{semiweakly simulates} $\phi$ if $h$ admits a block map section, that is, there exists a block map $\chi : Y \to (X, m\Z \times n\Z)$ with $h \circ \chi = \mathrm{id}_Y$. We say $\psi$ \emph{strongly simulates} $\phi$ if $h$ is bijective (and hence a conjugacy).
\end{definition}

Note that the leftmost arrow in the above diagram is an alphabet co-restriction of a blocking of $\psi$.
In particular, the domain of $h$ equals $\phi^{-1}(\tau(D^{\Z^2}))$.
The images $\tau(d) \in B^{m \times n}$ for $d \in D$ are sometimes called ``macrotiles''.

A weak simulation of $\phi$ by $\psi$ can be defined by the quadruple $(m,n,\tau,h)$.
In applications, we often want to produce this data algorithmically.
It is important to use the convention that to specify the partial block map $h$, it suffices to specify any block map $h' : (A^{\Z^2}, m\Z \times n\Z) \to C^{\Z^2}$ that agrees with $h$ on $\mathrm{dom}(h)$, and otherwise may produce configurations that are not even in $Y$.
In particular, we do not need to actually compute $\mathrm{dom}(h)$.

The full shift on the image side should be seen as an ``ambient full shift'' around the image that we are actually interested in. On the preimage side it is natural to use SFTs rather than full shifts, because alphabet co-restrictions of $\psi : X \to B^{\Z^2}$ can turn the domain into a proper SFT even if $X$ is a full shift.

We first prove that the simulation relations respect conjugacies and are transitive.

\begin{lemma}
  \label{lem:Conjugacy}
  Let $X \subset A^{\Z^2}$ be an SFT and $\psi : X \to B^{\Z^2}$ a block map.
  Let $\alpha : X \to Y$ be a conjugacy.
  Then $\psi$ strongly simulates $\psi \circ \alpha^{-1} : Y \to B^{\Z^2}$, and the simulation can be effectively computed from the local rules of $\psi$ and $\alpha$.
\end{lemma}

\begin{proof}
  We can choose $m = n = 1$, $\tau = \mathrm{id}_B$ and $h = \alpha$ in the definition of the strong simulation.
\end{proof}

\begin{lemma}
\label{lem:Composition}
Weak, semiweak and strong simulations are transitive, and the composition of two simulations can be effectively computed from its components.
\end{lemma}

\begin{proof}
  Suppose $X \subset A^{\Z^2}, Y \subset C^{\Z^2}, Z \subset E^{\Z^2}$. Suppose $\psi : X \to B^{\Z^2}$ weakly simulates $\phi : Y \to D^{\Z^2}$ with quadruple $(m, n, \tau, h)$ and $\phi : Y \to D^{\Z^2}$ weakly simulates $\chi : Z \to F^{\Z^2}$ with quadruple $(m', n', \tau', h')$. We claim that $\psi$ weakly simulates $\chi$ with quadruple $(mm', nn', \tau \circ \tau', h' \circ h)$ where $\tau \circ \tau' : F \to B^{mm' \times nn'}$ is a composition of substitutions, and $h' \circ h : (X, mm'\Z \times nn'\Z) \nrightarrow Z$ is a composition of partial functions.

Consider the following diagram of partial functions:
\begin{equation}
  \label{eq:transitive-diag}
  \begin{tikzcd}
    (X, mm'\Z \times nn'\Z) \arrow[negate]{r}{h} \arrow[swap,negate]{d}{\psi_{\tau(\tau'(F))}} & (Y, m'\Z \times n'\Z) \arrow[negate]{r}{h'} \arrow[swap,negate]{d}{\phi_{\tau'(F)}} & Z \arrow[swap]{d}{\chi} \\
    (B^{\Z^2}, mm'\Z \times nn'\Z) & (D^{\Z^2}, m'\Z \times n'\Z) \arrow[swap]{l}{\tau} & F^{\Z^2} \arrow[swap]{l}{\tau'}
  \end{tikzcd} 
\end{equation}
The right half is the diagram for the simulation of $\chi$ by $\phi$, which commutes by assumption.
Thus, it suffices to prove that the left half commutes.

If we apply the $(m', n')$-blocking operation to the diagram for the simulation of $\phi$ by $\psi$, we obtain the following commutative diagram:
\begin{equation}
  \label{eq:transitive-diag2}
  \begin{tikzcd}
    (X, mm'\Z \times nn'\Z) \arrow[negate]{r}{h} \arrow[swap,negate]{d}{\psi_{\tau(D)}} & (Y, m'\Z \times n'\Z) \arrow[swap]{d}{\phi} \\
    (B^{\Z^2}, mm'\Z \times nn'\Z) & (D^{\Z^2}, m'\Z \times n'\Z) \arrow[swap]{l}{\tau}
  \end{tikzcd}
\end{equation}
The only difference with this diagram and the left half of \eqref{eq:transitive-diag} is the further subalphabet restriction on both vertical arrows, so it suffices to verify that the domains of $\phi_{\tau(\tau'(F))}$ and $\tau \circ \phi_{\tau'(F)} \circ h$ are equal.

Suppose $x \in X$ is in the domain of $\psi_{\tau(\tau'(F))}$.
Since $\tau(\tau'(F)) \subset \tau(D)^{m' \times n'}$, in particular $x$ is in the domain of $\psi_{\tau(D)}$, which equals the domain of $\tau \circ \phi \circ h$ by the commutativity of \eqref{eq:transitive-diag2}.
We also have $\tau \phi h(x) = \psi(x) \in \tau(\tau'(F^{\Z^2}))$.
Because $\tau$ is injective, this implies $\phi h(x) \in \tau'(F^{\Z^2})$, and hence $x$ is in the domain of $\tau \circ \phi_{\tau'(F)} \circ h$.

Conversely, if $x$ is in the domain of $\tau \circ \phi_{\tau'(F)} \circ h$, then it is in the domain of $\tau \circ \phi \circ h = \psi_{\tau(D)}$ and $\psi(x) = \tau \phi h(x) \in \tau(\tau'(F^{\Z^2}))$.

For semiweak and strong universality we simply note that the properties of having a block map section and being a conjugacy are preserved by the composition operation.
\end{proof}

Strong simulation is the ``best'' we can hope for, in a sense. Semiweak simulation improves on weak simulation -- due to the factor map we inherit all complexity from $\psi$ (things like periodicity and asymptotic Kolmogorov complexity), and due to the section \emph{some} subset of the fibers contains no extra information.

\begin{definition}
  \label{def:universality}
If a block map $\psi : X \to B^{\Z^2}$ weakly (resp.\ strongly, semiweakly) simulates every block map $\phi : Y \to D^{\Z^2}$, where $X$ and $Y$ are SFTs, then we say that the triple $(\psi, X, B^{\Z^2})$ is \emph{weakly (resp.\ strongly, semiweakly) universal}. We say it is \emph{effectively weakly (resp.\ strongly) universal} if we can additionally compute a simulation $(m, n, \tau, h)$ from the local rule of $\phi$. 
For effective semiweak universality, we require that a section can also be computed.
\end{definition}

\begin{theorem}
  \label{thm:CircuitsStronglyUniversal}
  The circuit system $(f, Z_T, T^{\Z^2})$ is effectively strongly universal.
\end{theorem}

\begin{proof}
  Consider any $\phi : Y \to D^{\Z^2}$.
  By Lemmas~\ref{lem:Conjugacy} and~\ref{lem:Composition}, we may apply a recoding to $Y$ so that $\phi$ becomes a symbol map and $Y$ becomes the SFT of tilings by a set of Wang tiles (in particular the side colors determine the tile).
  In particular, we assume $Y \subset C^{\Z^2}$ where $C$ is a set of Wang tiles.

  We now construct macrotiles that simulate the symbols in $C$. Pick large (but otherwise arbitrary) $m, n$ and for each $c \in C$, in the corresponding macrotile $\tau(c) \in T^{m \times n}$ we draw wires from each border and connect them to the bottom of a computation zone in the middle of the macrotile (this part is independent of $c$). The wires therefore synchronize a sequence of bits between adjacent macrotiles, and the sequence on each side represents the color on the corresponding side of one of the Wang tiles in $C$. In the computation zone of the macrotile corresponding to $c$, we simply check that the bit sequences are the colors of a valid Wang tile that maps to $c$ under $\phi$: it is clear that we can draw any classical digital circuit with our tiles, and such circuits are computationally universal.


  The map $h$ simply maps each macrotile to the Wang tile in $C$ coded by the colors. If we perform computations in the computation zone deterministically, $h$ will be a bijection, so this is indeed a strong simulation.
\end{proof}

The following result states that the notion of effectiveness in Definition~\ref{def:universality} is in fact trivial.
This state of affairs is not as unnatural as it may at first appear.
For example, one usually defines $\Sigma^0_1$-completeness of a $\Sigma^0_1$ language $L$ as the condition that for every $\Sigma^0_1$ language $K$, there exists a computable reduction from $K$ to $L$.
This reduction is \emph{a priori} not uniformly computable from the description of $K$, but since there does exist a $\Sigma^0_1$-complete language for which it is uniformly computable (the halting problem), the same holds for every $\Sigma^0_1$-complete language.

A similar result is also proved in \cite[Theorem 11]{DeMaOlTh11} in the context of simulations between cellular automata.
The authors define axioms for weak and strong simulation relations, and show that if a strongly universal object exists, then every weakly universal object is strongly universal.

\begin{lemma}
\label{lem:DataIsComputable}
If a block map $\psi : C^{\Z^2} \to D^{\Z^2}$ is weakly, semiweakly or strongly universal, then it is effectively so.
\end{lemma}

\begin{proof}
  We observe that $\psi$ weakly (resp.\ strongly, semiweakly) simulates the circuit system $(f, Z_T, T^{\Z^2})$ with a fixed simulation $(m,n,\tau,h)$.
  Given any $\phi : Y \to D^{\Z^2}$, since $(f, Z_T, T^{\Z^2})$ is effectively strongly universal by Theorem~\ref{thm:CircuitsStronglyUniversal}, it simulates $\phi$ by some simulation $(m', n', \tau', h')$ that we can compute from a description of $\phi$ and $Y$. By Lemma~\ref{lem:Composition}, $(mm', nn', \tau \circ \tau', h' \circ h)$ is a simulation of $\phi$ by $\psi$ that can be computed effectively from $(m,n,\tau,h)$ and $(m', n', \tau', h')$.

  The simulation of $\phi$ by $(f, Z_T, T^{\Z^2})$ is strong (thus it is weak and semiweak), proving the weak and strong cases. For the semiweak case, we observe that since $h'$ is bijective, its inverse is a section, and inverting an invertible block map can be done effectively.
  By assumption we can also compute a section $\chi$ for $h$. Then $\chi \circ (h')^{-1}$ is a section for $h' \circ h$.
\end{proof}

\begin{theorem}
\label{thm:GoLSemiweaklyUniversal}
The Game of Life is semiweakly universal as a block map.
\end{theorem}

\begin{proof}
The proof of Theorem~\ref{thm:Sub} actually shows that the Game of Life semiweakly simulates the circuit system $(f, Z_T, T^{\Z^2})$. Namely we showed weak universality, by constructing the macrocells $\tau(T)$ and explained the map $h$ (which interprets the bits on our Game of Life wires as bits on the abstract wires carried in $Z_T$ on top of the gate tiles).

For semiweak universality we need a section for $h$. The macrocells corresponding to gate tiles were constructed so that we can pick zero borders away from the wires. For each gate tile $t$ and each tuple of bits $(b_1, b_2, b_3, b_4)$ carried on the wires, we pick for the macrotile corresponding to $t$ an arbitrary (but always the same) preimage for the macrotile where the Game of Life wires carry signals corresponding to the $b_i$, and the section of $h$ simply performs the corresponding substitution.
\end{proof}



\section{Corollaries}
\label{sec:Corollaries}

Here we list the theorems mentioned in the introduction, and explain how they are obtained. 

\begin{theorem}
\label{thm:FinitePreimageNPC}
Given a finite-support configuration, it is NP-complete whether it admits a $g$-preimage.
\end{theorem}


\begin{proof}
This problem is known to be in NP \cite[Theorem~1]{SaTo22}, so it suffices to prove NP-hardness.

Consider the circuit system $f : Z_T \to T^{\Z^2}$, and as usual let $X_T$ be the set of well-formed circuits and $Y_T = f(Z_T)$. From a given SAT instance $S$ we can easily compute a polynomial-size finite-support $X_T$-configuration $x_S$ which is in $Y_T$ if and only if the instance is satisfiable. Readers that skipped Section~\ref{sec:Universality} will note that the result immediately follows from the statement of Theorem~\ref{thm:Sub}.

To readers that did not skip Section~\ref{sec:Universality}, we explain how to deduce the result from semiweak universality by diagram-chasing. By Theorem~\ref{thm:GoLSemiweaklyUniversal}, there is a semiweak simulation $(m, n, \tau, h)$ of the circuit system by $g$ (indeed this is the simulation from Theorem~\ref{thm:Sub}).

We claim that $\tau(x_S)$ is in the image of $g$ if and only if the SAT instance $S$ is satisfiable. This is immediate from the simulation diagram \eqref{eq:Diagram} in this context:
\[
\begin{tikzcd}
(\{0,1\}^{\Z^2}, m\Z \times n\Z) \arrow[negate]{r}{h} \arrow[swap,negate]{d}{g_{\tau(T)}} & Z_T \arrow{d}{f} \\
(\{0,1\}^{\Z^2}, m\Z \times n\Z) & T^{\Z^2} \arrow[swap]{l}{\tau}
\end{tikzcd} 
\]
Namely, if $S$ is satisfiable, then $x_S \in T^{\Z^2}$ models a satisfiable circuit and has an $f$-preimage $z \in Z_T$. By the assumption that $h$ is surjective, we find an $h$-preimage $y \in \{0,1\}^{\Z^2}$ for $z$ such that $g(y) = \tau(x_S)$. In particular, $\tau(x_S)$ has a $g$-preimage.

Conversely, if $\tau(x_S)$ has a $g$-preimage $y \in \{0,1\}^{\Z^2}$, then by commutation we have $f(h(y)) = x_S$, showing $S$ is satisfiable.
\end{proof}

\begin{theorem}
\label{thm:Orphans}
  The set of orphans for $g$ is coNP-complete.
\end{theorem}

\begin{proof}
  Determining whether a finite pattern has a preimage is clearly in NP.
  As for hardness, if a rectangular pattern is padded with a thickness-4 border of dead cells and this padded pattern is not an orphan for $g$, then the extension to an infinite configuration by dead cells admits a $g$-preimage \cite[Theorem~3]{SaTo22}.
  By Theorem~\ref{thm:FinitePreimageNPC}, the latter condition is NP-hard.
  Alternatively, hardness can be proved directly as a finitary version of Theorem~\ref{thm:FinitePreimageNPC}.
\end{proof}

\begin{theorem}
\label{thm:PreimageOfPeriodic}
Given a subshift of finite type $Y$, we can effectively compute $m, n \in \N$ and an $m \times n$-periodic configuration $x \in \{0,1\}^{\Z^2}$ such that the system $(g^{-1}(x), m\Z \times n\Z)$ admits $Y$ as a split factor.
\end{theorem}

\begin{proof}
Consider the block map $\phi : Y \to \{0\}^{\Z^2}$ from $Y$ to a singleton subshift. 
By effective semiweak universality of $g$, we can compute a semiweak simulation $(m, n, \tau, h)$ of $\phi$ by $g$.
The substitution image $x = \tau(0^{\Z^2})$ is $m \times n$-periodic.

Consider now the diagram of the simulation:
\[
\begin{tikzcd}
(\{0,1\}^{\Z^2}, m\Z \times n\Z) \arrow[negate]{r}{h} \arrow[swap,negate]{d}{g_{\tau(0)}} & Y \arrow{d}{\phi} \\
(\{0,1\}^{\Z^2}, m\Z \times n\Z) & \{0\}^{\Z^2} \arrow[swap]{l}{\tau}
\end{tikzcd} 
\]
The domain of $g_{\tau(0)}$ is exactly $g^{-1}(x)$, and it also equals the domain of $h$.
Since the simulation is semiweak, $h$ admits a block map section $\chi : Y \to (g^{-1}(x), m\Z \times n\Z)$, which we can also effectively compute.
Hence $Y$ is a split factor of $(g^{-1}(x), m\Z \times n\Z)$, as required.
\end{proof}

The following theorems are immediate corollaries. 

\begin{theorem}
\label{thm:PeriodicPreimageUndecidable}
  It is undecidable (specifically, $\Pi^0_1$-complete) whether a given totally periodic configuration $x \in \{0,1\}^{\Z^2}$ has a $g$-preimage.
\end{theorem}

\begin{proof}
  Given an SFT $Y$, construct the configuration $x \in \{0,1\}^{\Z^2}$ as in Theorem~\ref{thm:PreimageOfPeriodic}.
  Now $x$ has a $g$-preimage if and only if $Y$ is nonempty.
  The latter was proved to be undecidable by Berger \cite{Be66}.
\end{proof}

\begin{theorem}
  \label{thm:PreimagePeriodicUndecidable}
  Given a totally periodic configuration $x \in \{0,1\}^{\Z^2}$ that has a preimage under $g$, it is undecidable whether it has a totally periodic $g$-preimage.
  Given a totally periodic $x$, it is $\Sigma^0_1$-complete whether it has a totally periodic $g$-preimage.
\end{theorem}

\begin{proof}
  Given an SFT $Y$, construct the configuration $x \in \{0,1\}^{\Z^2}$ as in Theorem~\ref{thm:PreimageOfPeriodic}.
  If $Y$ is nonempty, then the corresponding periodic configuration $x$ has a preimage.
  Furthermore, if $Y$ has a periodic point, then by semiweak universality so does $x$ (by applying the section).
  Conversely, any periodic preimage of $x$ maps to a periodic point of $Y$.
  It follows that $Y$ has a periodic point if and only if $x$ has a periodic preimage.
  Given a nonempty SFT, it is $\Sigma^0_1$-hard, in particular undecidable, whether it contains a totally periodic point \cite[Theorem~5]{Je10}, proving the first claim.

  The second claim follows from the same $\Sigma^0_1$-hardness result, and the observation that the problem is $\Sigma^0_1$ since we simply need to exhibit a totally periodic preimage to prove that one exists.
\end{proof}

We note that Theorem~\ref{thm:PreimagePeriodicUndecidable} does not follow directly from weak universality, since we cannot prove the existence of a totally periodic preimage without the section.
More concretely, take any aperiodic SFT $X$ and consider the block map $h : X \times \{0,1\}^{\Z^2} \to \{0,1\}^{\Z^2}$ defined by $h(x, y) = g(y)$.
Then $h$ is weakly universal, but no configuration has a totally periodic $h$-preimage, so the problems of Theorem~\ref{thm:PreimagePeriodicUndecidable} are trivially decidable.

\begin{theorem}
\label{thm:ArecursiveOnly}
  There exists a totally periodic configuration $x \in \{0,1\}^{\Z^2}$ that has a $g$-preimage but no computable $g$-preimage.
\end{theorem}

\begin{proof}
  There exists an SFT $Y$ with no computable configurations \cite{My74}.
  Construct the configuration $x$ as in Theorem~\ref{thm:PreimageOfPeriodic}, and observe that $x$ cannot have a computable $g$-preimage, since its $h$-image would be a computable configuration in $Y$.
\end{proof}

\begin{theorem}
\label{thm:Kolmogorov}
There exists a totally periodic configuration that has a $g$-preimage, but every preimage has $\Omega(n)$ Kolmogorov complexity in every $n \times n$-pattern.
\end{theorem}

\begin{proof}
  The proof is similar to the ones above, but uses the existence of SFTs with maximal Kolmogorov complexity in all patterns \cite{DuLeSh08}.
\end{proof}

Of course, the Kolmogorov complexity of the $g$-preimages is much lower than that of the factor -- the hidden constant in the $\Omega(n)$ is very small.




%

\section{Implementation and verification of SAT-based searches}
\label{sec:Implementation}

We constructed the gadgets of Section~\ref{sec:BasicGadgets} using two different SAT-based search programs and some manual work.
Both programs are based on the following idea.
Fix a finite set $D \subset \Z^2$ of cells and a set $\mathcal{F} \subset \{0,1\}^D$ of possible ``forced preimages''.
Our goal is to produce a single finite pattern $P$ such that the set $g^{-1}(P)|_D = \{ Q|_D \mid Q \in g^{-1}(P) \}$ of shape-$D$ subpatterns of its preimages is exactly $\mathcal{F}$.
We see this as an optimization problem: there is a hard constraint $\mathcal{F} \subset g^{-1}(P)|_D$, and the goal is to minimize the cardinality of $g^{-1}(P)|_D \setminus \mathcal{F}$, hopefully to zero.
We also need to impose additional hard constraints, such as $P$ containing a wire at a certain position and 0-cells on the rest of its border, or relaxations, like only considering preimages that have a valid ``incoming'' signal on top of a wire of $P$.

The first algorithm is essentially a backtracking hill-climbing optimizer, and the second one is a genetic algorithm.
Both programs are written in Python and use the PySat library \cite{pysat} to invoke the Glucose 4.1 SAT solver \cite{glucose41}.
The source code for the hill-climber is available at \cite{hill-climber}.

\subsection{The hill-climber}

The hill-climber program constructs the pattern $P$ iteratively, starting from the empty pattern and adding one or more specified cells at a time.
It expects the following parameters:
\begin{itemize}
\item
  A finite set $\{D_1, \ldots, D_n\}$ of nonempty finite subsets of $\Z^2$.
  Their union $D = \bigcup_{i=1}^n D_i$ is the domain of the forced preimages.
\item
  A finite set $\{q_1, \ldots, q_k\}$ of finite or infinite patterns, given as functions $q_i : \Z^2 \to \{0,1,\bot\}$ with $\bot$ meaning an unspecified cell.
\item
  For each $i \in \{1, \ldots, k\}$, a set $\mathcal{F}_i \subset \{0,1\}^D$ of forced patterns.
  The semantics is that restricting to $D$ those preimages that are compatible with $q_i$ should produce exactly the set $\mathcal{F}_i$.
\item
  A finite or infinite pattern, again given as a function $p : \Z^2 \to \{0,1,\bot\}$.
  It represents a constraint on $P$: if $p(\vec v) = b \in \{0,1\}$, then we must have $P_{\vec v} = b$.
\item
  Several technical parameters that guide the search process.
\end{itemize}

The program proceeds in rounds.
On each round, we enumerate the outer border $\partial P$ of the current candidate pattern $P$, which consists of those cells that are not in the domain of $P$ but have an immediate neighbor (orthogonal or diagonal) that is.
For each $\vec v \in \partial P$ and $b \in \{0,1\}$, we extend $P$ into a new pattern $P'$ by specifying $P'_{\vec v} = b$, and compute its \emph{score} $s(P')$ as defined below.
The extension with the lowest score will replace $P$.
If no extension has a strictly lower score than $P$, then we enumerate all extensions by two adjacent cells on the border, then three, and so on up to some bound.
If no better extension is still found, we repeatedly try to extend $P$ by randomly specifying the values of a randomly chosen rectangle of cells near the border of $P$.
After some number of tries, we give up and backtrack, choosing the next best extension of the previous round instead.
The program keeps track of the entire history of the search, so it can backtrack as far as needed.

The score of a candidate pattern $P$ is defined as
\[
  s(P) = \sum_{i=1}^k \sum_{j=1}^n \frac{1}{|D_j|} \log \left( 1 + \sum_{R \in \mathcal{Q}_{(i,j)}} \frac{M(i,j,R)}{|D_j|} \right)
\]
if it is \emph{valid} (see below), and $s(P) = \infty$ if it is invalid, with lower score being better.
Here, $\mathcal{Q}_{(i,j)}$ is the set of patterns $R = Q|_{D_j}$ where $Q \in g^{-1}(P)$ is a preimage that agrees with $q_i$, minus those in $(\mathcal{F}_i)|_{D_j}$, and
\[
  M(i,j,R) = 1 + \max_{Q \in \mathcal{F}_i} |\{ \vec v \in D_j \mid R_{\vec v} = Q_{\vec v} \}|.
\]
is one plus the largest number of cells on which $R$ agrees with some pattern of $\mathcal{F}_i$.

The high-level idea of the score function is the following.
For each $i \in \{1, \ldots, k\}$, $j \in \{1, \ldots, n\}$ and $R \in \{0,1\}^{D_j} \setminus \mathcal{F}_i$, we want to encourage the program to extend the pattern $P$ in such a way that no preimage $Q \in g^{-1}(P)$ that is compatible with $q_i$ contains the pattern $R$ at $D_j$.
Each pattern $R$ that does occur in such a preimage incurs a penalty proportional to $M(i,j,R) / |D_j| > 0$ inside the logarithm, so the program is incentivized to extend $P$ in such a way that these patterns do not occur in its preimages.
Furthermore, the program is encouraged to first get rid of patterns $R$ with higher $M(i,j,R)$, that is, those that are similar to some pattern in $\mathcal{F}_i$ (which we do not want to forbid).
Our intuition is that such patterns $R$ are easier to handle when the pattern $P$ is still small and newly added cells are close to the domain $D_j$.

Notice that since we have covered the domain $D$ by the subsets $D_1, \ldots, D_k$ that are scored separately, having a score of $0$ does not guarantee the property that restricting $q_i$-compatible preimages of $P$ to $D$ gives exactly $\mathcal{F}_i$, unless $k = 1$.
The reason for the cover is efficiency: the size of each $\mathcal{Q}_{(i,j)}$ is initially exponential in that of $D_j$.
However, as $P$ grows, the subsets $\mathcal{Q}_{(i,j)}$ shrink, and when two of them become small enough, the program replaces the respective sets $D_{j(1)}$ and $D_{j(2)}$ with their union, thus decreasing $k$.

Let us now discuss the validity of patterns and the ways to verify it.
There are three reasons why a pattern $P$ can be invalid.
The first is that for some $\vec v \in \Z^2$ we have $p(\vec v) = b \in \{0,1\}$ and $P_{\vec v} = 1-b$.
This can be verified immediately after extending $P$ by the cell $\vec v$.
The second is that for some $i \in \{1, \ldots, k\}$ and $Q \in \mathcal{F}_i$, there is no preimage $R \in g^{-1}(P)$ that is compatible with $q_i$ and satisfies $R|_D = Q|_D$.
This can be verified with a single call to a SAT-solver: it is simple to construct a SAT instance that is unsatisfiable if and only if the above property holds.

The third reason is more subtle.
Consider the situation that for some $i \in \{1, \ldots, n\}$ there are two preimages $Q, R \in g^{-1}(P)$ which are compatible with $q_i$, such that $Q|_D \in \mathcal{F}_i$, $R|_D \notin \mathcal{F}_i$, and $Q_{\vec v} = R_{\vec v}$ for every cell $\vec v$ near the border of $P$.
We call this pair of preimages a \emph{diamond}.
Then for every extension $P'$ of $P$, either both $Q$ and $R$ or neither of them extend to $q_i$-compatible preimages of $P'$.
Unless $Q|_D$ has another preimage that is not part of a diamond, the pattern $R|_D$ can never be forbidden from $q_i$-compatible preimages by extending $P$.
We deem $P$ to be invalid if it admits a diamond, as this potentially dangerous property can be checked with a bounded number of calls to a SAT-solver, whereas there does not seem to be an efficient way of determining whether each forced pattern admits a preimage that is not part of a diamond.

The final zero-score pattern $P$ typically has an irregular shape, so it is completed into a rectangle one cell at a time, favoring 0-cells unless doing so would result in an invalid pattern.
The constraint of admitting preimages with only 0-cells near the border is not enforced by the program, so it must be done manually.
We added and removed 1-cells near wires in a trial-and-error fashion until the property held.

The program has various additional features that we have not fully described here.
We only mention a forced preference to filling small crevices on the border of $P$ before other cells, the possibility of forcing $P$ to be periodic in some direction, and support for multithreading.

The hill-climber program was used to find all patterns of Section~\ref{sec:BasicGadgets} except the charger.\footnote{In fact, while we did not manage to directly charge a wire with the hill-climber, we did find a ``relaxer'' using it. It has the specs $N \vdash E; EN \in \{00, 10\}$ meaning given a charged wire in a fixed phase, it can relax the phase to one of two choices. In combination with the enforcer this gives a (rather large) charger gadget.}
We note that the patterns were found with various early iterations of the program with different scoring functions.
Together with the random component of the search, this means that the current program may produce vastly different gadgets.

\subsection{The genetic algorithm}

Like the hill-climbing algorithm, the genetic algorithm attempts to find a pattern that forces certain behavior for the preimage. Unlike the hill-climbing algorithm, the genetic algorithm was not written as a general-purpose search program, but rather experimentally with hard-coded inputs and parameters. Since we only found the charger with this approach, we concentrate on the charger-specific optimization problem. As finding patterns with it requires extensive human effort, we have not included its source code.

What we want is to find a pattern $P \in \{0,1\}^{[0,2n-1] \times [0,m]}$ satisfying the following conditions.
The first two rows contain only a wire, which is precisely at the midpoint, i.e.\ $\supp(P|_{[0,2n-1] \times [0,1]}) = [n-1,n] \times [0,1]$, and the remainder of the border of thickness $2$ contains only zeroes.
For at least two of the wire signals that can appear on $[n-2,n+1] \times [0,1]$, they appear with zeroes on the remainder of the border.
Conversely, the number of different restrictions of preimages to $[n-2,n+1] \times [0,1]$ is precisely two, namely we see precisely two different phases of the wire signal.

The conditions on the border of $P$ hold at all times by our choices of initial patterns, and choice of how we modify them. We use a simple scoring function to check how far a pattern is from having the latter properties: If only zero or one wire signal phases occur in the preimages with zeroes on the border, the pattern is immediately discarded (or given score $-\infty$). Then we simply count the number $s$ of restrictions of preimages to $[n-2,n+1] \times [0,2]$, and use the score $s-2$; the maximal possible score is then $0$, and it means that there are only two possibilities, which must be two wire signal phases, since we are guaranteed that two phases are even possible with zeroes elsewhere on the border of the preimage pattern. Both the check and the counting are easy to perform with a SAT solver.

We are then in the standard scheme of searching for a binary pattern under a scoring function, with a particular score indicating a pattern with the desired property. To solve this problem genetically, the algorithm keeps, at all times, track of a finite number $n$ of patterns (we used an upper bound of 100). It performs a given number of iterations on this population. In each iteration, it performs reproduction with genetic modifications to produce a number of new patterns. They are scored, and the best $n$ are kept, with a minor twist: to keep the population from becoming too homogenous, we additionally compare ``genetic similarity'' (by simply counting the number of equal cells!), and do not include a pattern if it is too similar to multiple already included patterns.

Our reproduction rules are the following.
For cloning and mutation, we pick a pattern at random, then pick a small set of random small rectangles contained in its domain, and randomly pick the new contents of these rectangles.
For crossover, we pick two random patterns $P_1, P_2$ in our set, and construct a new pattern $P$ by picking the $i$th row randomly from the corresponding row of $P_1$ or $P_2$, for each $i \in [0,m]$. The choices of which pattern each row is taken from are not independent, but are generated by a Markov process that prefers to pick the row $i+1$ from the same pattern that row $i$ is copied from.

The idea of crossover is that some rows might somehow eliminate some preimages, and other rows might eliminate others. There were indeed situations during search where the algorithm was stuck for a while, after which crossover finally found a better solution. We do not know, however, whether this was truly a result of combining two patterns, or whether it could have been replaced by a more aggressive mutation rule.

\subsection{Verification}

The repository \cite{pat-repo} contains a Python script that verifies the claimed properties of the gadgets of Section~\ref{sec:BasicGadgets} (that they implement the claimed specs), and those gadgets of Section~\ref{sec:CompositeGadgets} that are small enough for the SAT-solver to handle (that they implement the claimed specs and have the correct dimensions and wire offsets).
The script allows the user to choose between three different SAT encodings of the local rule of the Game of Life (a divide-and-conquer encoding with no auxiliary variables, an encoding based on sorting networks, and an encoding based on a merge operation) to reduce the likelihood of errors in the encoding resulting in false positives.

\section{Conclusion}

The contributions of this article have many forms.
First, of course, we have shown the NP-completeness of the existence of preimages of finite patterns for the Game of Life, and proved several hardness results for the preimages of infinite periodic configurations.
Second, we present a unified theoretical framework for proving such computational hardness results about finding preimages under a CA.
Third, we provide an implementation of an algorithm for constructing gadgets that can in principle be adapted to any two-dimensional CA, and with a bit more work, to any CA in any dimension greater than 1.

In Section~\ref{sec:JeandelRao} we produced an explicit $6210 \times 37800$-periodic configuration that admits preimages, but only aperiodic ones.
We have not seriously tried to optimize the periods, apart from seeking to produce a somewhat readable picture in Figure~\ref{fig:JeandelRao}.

\begin{question}
  What are the smallest periods (for example in terms of the size of the fundamental domain) of a periodic configuration that has a $g$-preimage, but no periodic preimage?
\end{question}

In Section~\ref{sec:Universality} we developed a theory of simulations between block maps and proved that the Game of Life is semiweakly universal.
We leave open whether it is strongly universal.
In concrete terms, strong universality corresponds to simulating the circuit system of Section~\ref{sec:CompositeGadgets} using rectangular patterns in such a way that the preimages of a configuration that simulates a circuit correspond bijectively to the valid assignments of signals in the circuit: there must not be any extraneous preimages.
This would imply, for example, that it is undecidable whether a periodic configuration that admits a preimage also admits an aperiodic preimage, which for now we cannot conclude.

\begin{question}
  Is the Game of Life strongly universal as a block map?
\end{question}

While our results are particular to a single step of Life, the method of finding constrained preimages by computer searches accompanied by SAT-solvers is in principle applicable to any cellular automaton in at least two dimensions (and any fixed number of iterations of such a CA).
However, the methods do not readily extend to beyond a fixed number of iterations, and even then the amount of computational resources needed grows rapidly with the radius of the CA.

Of particular interest would be a general method for performing \emph{long-term} backwards computation in the Game of Life. So far, we do not have such tools. Namely, while we show that in a single step backwards in time, a configuration can enforce (even infinitary) computation on every preimage, there is nevertheless a lot of freedom in the choice of the preimage, and thus we have no control on the second-order preimage.

\begin{question}
Can similar results be proved for higher powers of the Game of Life?
\end{question}

\begin{question}
Are there Game of Life configurations such that all of their preimage chains perform universal computation backwards in time?
\end{question}

\bibliographystyle{plain}
\bibliography{golfin1bib}{}

\end{document}